\documentclass[a4paper,USenglish]{lipics}
 
\usepackage{microtype}
\title{Towards Establishing Monotonic Searchability in Self-Stabilizing Data Structures\footnote{This work was partially supported by the German Research Foundation (DFG) within the Collaborative Research Center ``On-The-Fly Computing'' (SFB 901).}}

\author[1]{Christian Scheideler}
\author[2]{Alexander Setzer}
\author[3]{Thim Strothmann}
\affil[1]{Paderborn University\\
  Fürstenallee 11, Paderborn, Germany}
\affil[2]{Paderborn University\\
  Fürstenallee 11, Paderborn, Germany}
\affil[3]{Paderborn University\\
  Fürstenallee 11, Paderborn, Germany}
\authorrunning{C. Scheideler and A. Setzer and T. Strothmann}

\Copyright{Christian Scheideler, Alexander Setzer, and Thim Strothmann}


\usepackage{xcolor}
\usepackage{xspace}
\usepackage{algorithm}
\usepackage[noend]{algpseudocode}
\usepackage{xifthen}
\usepackage{tikz}

\newcommand\mytodo[1]{\textcolor{red}{TODO\ifthenelse{\isempty{#1}}{}{:} #1}}

\algrenewcommand{\algorithmiccomment}[1]{$\rhd$ #1}
\algrenewcommand\algorithmicprocedure{\textbf{action}}

\newcommand{\blp}{\textsc{Build-List+}\xspace}
\newcommand{\srp}{\textsc{Search+}\xspace}
\newcommand{\srpwithoutxspace}{\textsc{Search+}}
\newcommand{\blpp}{\textsc{Build-List*}\xspace}
\newcommand{\srpp}{\textsc{Search*}\xspace}

\newcommand{\linearize}[1]{\textsc{Linearize(\ensuremath{#1})}\xspace}
\newcommand{\introduce}[1]{\textsc{Introduce(\ensuremath{#1})}\xspace}
\newcommand{\tempdelegate}[1]{\textsc{TempDelegate(\ensuremath{#1})}\xspace}
\newcommand{\timeout}{\textsc{Timeout}\xspace}
\newcommand{\search}[1]{\textsc{Search(\ensuremath{#1})}\xspace}
\newcommand{\initsearch}[1]{\textsc{InitiateNewSearch(\ensuremath{#1})}\xspace}
\newcommand{\forwardprobe}[1]{\textsc{ForwardProbe(\ensuremath{#1})}\xspace}

\newcommand{\psuccess}[1]{\textsc{ProbeSuccess(\ensuremath{#1})}\xspace}
\newcommand{\pfail}[1]{\textsc{ProbeFail(\ensuremath{#1})}\xspace}

\newcommand{\rsp}{\ensuremath{R_s^+}\xspace}
\newcommand{\lsp}{\ensuremath{L_s^+}\xspace}
\newcommand{\rsvw}{\ensuremath{R_s(v,w)}\xspace}

\newcommand{\fdp}{$\mathcal{FDP}$\xspace}
\newcommand{\nidec}{$\mathcal{NIDEC}$\xspace}

\newcommand{\revandlin}[1]{\textsc{ReverseAndLinearize(\ensuremath{#1})}\xspace} 
\newcommand{\revandlinREQ}[1]{\textsc{ReverseAndLinearizeREQ(#1)}\xspace}
\newcommand{\revandlinACK}[1]{\textsc{ReverseAndLinearizeACK(#1)}\xspace} 

\newcommand{\templeft}[1][]{\ensuremath{Temp_{L}\ifthenelse{\isempty{#1}}{}{(#1)}}\xspace}
\newcommand{\tempright}[1][]{\ensuremath{Temp_{R}\ifthenelse{\isempty{#1}}{}{(#1)}}\xspace}

\begin{document}

\maketitle

\begin{abstract}
Distributed applications are commonly based on overlay networks interconnecting their sites so that they can exchange information. 
For these overlay networks to preserve their functionality, they should be able to recover from various problems like membership changes or faults. 
Various self-stabilizing overlay networks have already been proposed in recent years, which have the advantage of being able to recover from any illegal state, but none of these networks can give any guarantees on its functionality while the recovery process is going on. 
We initiate research on overlay networks that are not only self-stabilizing but that also ensure that searchability is maintained while the recovery process is going on, as long as there are no corrupted messages in the system. 
More precisely, once a search message from node $u$ to another node $v$ is successfully delivered, all future search messages from $u$ to $v$ succeed as well. 
We call this property {\em monotonic searchability}.
 We show that in general it is impossible to provide monotonic searchability if corrupted messages are present in the system, which justifies the restriction to system states without corrupted messages. 
 Furthermore, we provide a self-stabilizing protocol for the line for which we can also show monotonic searchability. 
 It turns out that even for the line it is non-trivial to achieve this property. 
 Additionally, we extend our protocol to deal with node departures in terms of the Finite Departure Problem of Foreback et. al (SSS 2014). 
 This makes our protocol even capable of handling node dynamics.
 \end{abstract}

\section{Introduction}
The Internet has opened up tremendous opportunities for people to interact and exchange information.
 Particularly popular ways to interact are peer-to-peer systems and social networks. 
 For these systems to stay popular, it is very important that they are highly available. 
 However, once these systems become large enough, changes and faults are not an exception but the rule.
Therefore, mechanisms are needed that ensure that whenever there are problems, they are quickly repaired, and all parts of the system that are still functional should not be affected by the repair process. 
Protocols that are able to recover from arbitrary states are also known as \emph{self-stabilizing} protocols.

Since the seminal paper of Dijkstra in 1974~\cite{Dijkstra74}, self-stabilizing protocols have been investigated for many classical problems including leader election, consensus, matching, clock synchronization and token distribution problems.
Recently, also various protocols for self-stabilizing overlay networks have been proposed (e.g., \cite{corona,JRSST09,DolevT2013,JacobRSS2012,DolevK08, AspnesW07,KniesburgesKS12,rechord,DBLP:journals/tcs/BernsGP13}). 
However, for all of these protocols it is only known that they \emph{eventually} converge to the desired solution, but the convergence process is not necessarily \emph{monotonic}. 
In other words, it is not ensured for two points in time $t, t'$ with $t<t'$ that the functionality of the topology at time $t'$ is better than the functionality at time $t$.

In this paper, we focus on protocols for self-stabilizing overlay networks that guarantee the \emph{monotonic} preservation of a characteristic that we call \emph{searchability}, i.e., once a search message from node $u$ to another node $v$ is successfully delivered, all future search messages from $u$ to $v$ succeed as well. 
Searchability is a useful and natural characteristic for an overlay network since searching for other participants is one of the most common tasks in real-world networks. 
Moreover, a protocol that preserves monotonic searchability has the huge advantage that in every state, even if the self-stabilization process has not converged yet, the already built topology can already be used for search requests.

As a starting point for rigorous research on monotonic searchability, we will focus on building a self-stabilizing protocol that preserves monotonic searchability for the line graph. 
Although the topology itself is fairly simple, to preserve searchability during the self-stabilization process turns out to be quite challenging. 
Additionally, we study monotonic searchability for the line graph if the node set is dynamic, i.e., nodes are allowed to leave the network.

\subsection{Model}
We consider a distributed system consisting of a fixed set of nodes in which each node has a unique reference and a unique immutable numerical identifier (or short id).
The system is controlled by a protocol that specifies the variables and actions that are available in each node. 
In addition to the protocol-based variables there is a system-based variable for each node called \emph{channel} whose values are sets of messages. 
We denote the channel of node $u$ as $u.Ch$ and $u.Ch$ contains all incoming messages to $u$. 
Its message capacity is unbounded and messages never get lost.
A node can add a message to $u.Ch$ if it has a reference to $u$.
Besides these channels there are no further communication means, so only point-to-point communication is possible.

There are two types of actions. 
The first type of \emph{action} has the form of a standard procedure 
$\langle label\rangle (\langle parameters \rangle): \langle
command \rangle$, where $label$ is the unique name of that action,
$parameters$ specifies the parameter list of the action, and $command$
specifies the statements to be executed when calling that action. Such actions
can be called remotely. In fact, we assume that every message must be of the
form $\langle label \rangle (\langle parameters \rangle)$ where $label$
specifies the action to be called in the receiving node and $parameters$
contains the parameters to be passed to that action call. All other messages
will be ignored by the nodes. Apart from being triggered by messages,
these actions may also be called locally by the nodes, which causes their
immediate execution. 
The second type of action has the form $ \langle
label\rangle: \langle guard \rangle \longrightarrow \langle command \rangle$,
where $label$ and $command$ are defined as above and $guard$ is a predicate
over local variables. We call an action whose guard is simply \textbf{true} a
\emph{timeout} action.

The \emph{system state} is an assignment of a value to every variable of each node and messages to each channel. 
An action in some node $p$ is \emph{enabled} in some system state if its guard evaluates to \textbf{true}, or if there is a message in $p.Ch$ requesting to call it.
In the latter case the corresponding message is processed (in which case it is removed from $p.Ch$).
An action is \emph{disabled} otherwise. 
Receiving and processing a message is considered as an atomic step.

A \emph{computation} is an infinite fair sequence of system states such that
for each state $s_i$, the next state $s_{i+1}$ is obtained by executing an
action that is enabled in $s_i$. This disallows the overlap of action
execution. That is, action execution is \emph{atomic}. We assume \emph{weakly
fair action execution} and \emph{fair message receipt}. Weakly fair action
execution means that if an action is enabled in all but finitely many states
of the computation, then this action is executed infinitely
often. Note that the timeout action of a node is executed infinitely
often. Fair message receipt means that if the computation contains a state
where there is a message in a channel of a node that
enables an action in that node, then that action is eventually executed
with the parameters of that message, i.e., the message is eventually
processed. Besides these fairness assumptions, we place no bounds on message
propagation delay or relative nodes execution speeds, i.e., we allow fully
asynchronous computations and non-FIFO message delivery.
A \emph{computation suffix} is a sequence of computation states past a particular state of this computation. 
In other words, the suffix of the computation is obtained by removing the initial state and finitely
many subsequent states. 
Note that a computation suffix is also a computation.

We consider protocols that do not manipulate the internals of node
references. Specifically, a protocol is \emph{compare-store-send} if the only
operations that it executes on node references is comparing them, storing
them in local memory and sending them in a message. That is, operations on
references such as addition, radix computation, hashing, etc. are not used. In
a compare-store-send protocol, if a node does not store a reference in its
local memory, the node may learn this reference only by receiving it in a
message. A compare-store-send protocol cannot introduce new references to the
system. It can only operate on the references that are already there.

The overlay network of a set of nodes is determined by their knowledge of each other. 
We say that there is a (directed) \emph{edge} from $a$ to $b$, denoted by $(a,b)$, if node $a$ stores a reference of $b$ in its local memory or has a message in $a.Ch$ carrying the reference of $b$. 
In the former case, the edge is called \emph{explicit} (drawn solid in figures), and in the latter case, the edge is called \emph{implicit} (drawn dashed). 
With $NG$ we denote the directed \emph{network (multi-)graph} given by the explicit and implicit edges.
$ENG$ is the subgraph of $NG$ induced by only the explicit edges. 
A \emph{weakly connected component} of a directed graph $G$ is a subgraph of $G$ of maximum size so that for any two nodes $u$ and $v$ in that subgraph there is a (not necessarily directed) path from $u$ to $v$. 
Two nodes that are not in the same weakly connected component are \emph{disconnected}.
We say a node $a$ is to the \emph{left} (\emph{right}, respectively) of a node $b$ if $id(a)<id(b)$ ($id(a)>id(b)$).
If there is an edge $(a,b)$ between the two, then $a$ is a \emph{left neighbor} (\emph{right neighbor}).
For three nodes $a,b,c$ with $id(a)<id(b), id(a)<id(c)$ (or $id(a)>id(b), id(a)>id(c)$, respectively), we say a node $b$ is \emph{closer} to $a$ than $c$, if $\vert id(a)-id(b) \vert < \vert id(a)-id(c) \vert$.
If it is clear from the context we sometimes refer to the identifier of a node by dropping the $id$ notation to , e.g., we write $a < b$ instead of $id(a)<id(b)$.

In this paper we are particularly concerned with search requests, i.e., \search{v,destID} messages that are routed along $ENG$ according to a given routing protocol, where $v$ is the sender of the message and $destID$ is the identifier of a node we are looking for.
Note that $destID$ does not necessarily belong to an existing node $w$, since we also want to model search requests to not existing nodes. 
If a \search{v,destID} message reaches a node $w$ with $id(w)=destID$, the search request \emph{succeeds}; if the message reaches some node $u$ with $id(u) \ne destID$ and cannot be forwarded anymore according to the given routing protocol, the search request \emph{fails}.
We assume that nodes themselves initiate \search{} requests at will.
Therefore, the \search{destID} action is never explicitly called.

We need some additional notation for our results of Section~\ref{sec:theBLPPAlgorithm}, in which we extend the protocol to handle nodes that want to leave the system.
A node $u$ has a variable $\emph{mode} \in \{ \text{leaving}, \text{staying} \}$ that is read-only. 
If this variable is set to \textbf{leaving}, the node is \emph{leaving}; the node is \emph{staying} if the variable is set to \textbf{staying}.
Note that staying nodes can dynamically decide at any arbitrary state if they want to leave the system by executing a corresponding \emph{leave action}. 
However, a leaving node cannot switch back to staying.
The ultimate goal of a leaving node is to depart from the system.
There is one special command that is important for the study of leaving nodes: \textbf{exit}. 
If a node executes \textbf{exit} it enters a designated \emph{exit state} and all remaining edges to or from that node are deleted. 
We call such a node \emph{gone}. 
A node that is not gone is called \emph{present}.
For a gone node all actions are disabled, in particular it will not execute the timeout action regularly.

\subsection{Problem Statement}

A protocol is \emph{self-stabilizing} if it satisfies the following two properties.
\begin{description}
\item[\emph{Convergence:}] starting from an arbitrary system state, the protocol is guaranteed to arrive at a legitimate state.
\item[\emph{Closure:}] starting from a legitimate state the protocol remains in legitimate states thereafter.
\end{description}
A self-stabilizing protocol is thus able to recover from transient faults regardless of their nature. 
Moreover, a self-stabilizing protocol does not have to be initialized as it eventually starts to behave correctly regardless of its initial state. 
In \emph{topological self-stabilization} we allow self-stabilizing protocols to perform changes to the overlay network, resp.~$NG$. 
A legitimate state may then include a particular graph topology or a family of graph topologies.

In this paper we want to build a self-stabilizing protocol for the \emph{linearization problem}, i.e., the nodes are sorted by identifiers and each node stores only two references: its closest successor and its closest predecessor.
From a global point of view, the nodes build a \emph{line graph} topology.
Of course, searching is easy once a legitimate state has been reached.
However, searching reliably during the stabilization phase is much more involved.
We say a (self-stabilizing) protocol satisfies \emph{monotonic searchability} according to some routing protocol $R$ if it holds for any pair of nodes $v,w$ that once a \search{v,id(w)} request (that is routed according to $R$) initiated at time $t$ succeeds, any \search{v,id(w)} request initiated at a time $t' > t$ will succeed.
We do not mention $R$ if it is clear from the context.
A protocol is said to satisfy \emph{non-trivial} monotonic searchability if it satisfies monotonic searchability and in every computation of the protocol there is a suffix such that for each pair of nodes $v,w$ for which there is a path from $v$ to $w$ in the target topology \search{v,id(w)} requests will succeed.

Furthermore, we give a self-stabilizing protocol that satisfies non-trivial monotonic searchability, solves the linearization problem and solves the \emph{Finite Departure Problem} of~\cite{departure1}. The following problem statement is adapted from~\cite{KoutsopoulosSS15}:
\begin{description}
\item[\emph{Finite Departure Problem} (\fdp)]: In case the \textbf{exit} command is available, eventually reach a system state in which (i) every staying node is awake, (ii) every leaving node is gone and (iii) for each weakly connected component of the initial network graph, the staying nodes in that component still form a weakly connected component.
\end{description}
Consequently, a leaving node $u$ should \emph{safely} execute \textbf{exit}, i.e., the removal of $u$ and its incident edges from $NG$ does not disconnect any present nodes and does not violate searchability.

\subsection{Related work}

The idea of self-stabilization in distributed computing was introduced in a classical paper by E.W. Dijkstra in 1974~\cite{Dijkstra74}, in which he looked at the problem of self-stabilization in a token ring. 
In order to recover certain network topologies from any weakly connected state, researchers started with simple line and ring networks (e.g.~\cite{ShakerR05,self-stabilizing-list,self-stabilizing-list2}.
Over the years more and more network topologies were considered, ranging from skip lists and skip graphs~\cite{corona,JRSST09}, to expanders~\cite{DolevT2013}, Delaunay graphs~\cite{JacobRSS2012}, hypertrees and double-headed radix trees~\cite{DolevK08, AspnesW07}, small-world graphs~\cite{KniesburgesKS12} and a Chord variant~\cite{rechord}. 
Also a universal algorithm for topological self-stabilization is known~\cite{DBLP:journals/tcs/BernsGP13}.

Close to our work is the notion of \emph{monotonic convergence} by Yamauchi and Tixeuil~\cite{YamauchiT10}. A self-stabilizing protocol is monotonically converging if every change done by a node $p$ makes the system approach a legitimate state and if every node changes its output only once. 
The authors investigate monotonically converging protocols for different classic distributed problems (e.g., leader election and vertex coloring) and focus on the amount of non-local information that is needed for them.

Our study of the \emph{Finite Departure Problem} is heavily inspired by~\cite{departure1}, in which the authors propose the aforementioned problem to study graceful departures of nodes in a self-stabilizing setting, i.e., nodes that want to leave a distributed system should decide when they can leave without affecting weak connectivity of the topology. 
They conclude that in general it is not possible to solve the \fdp. 
However, with the use of distributed oracles (which are specialized failure detectors~\cite{ChandraT96}) the authors propose a protocol that solves the problem and arranges the nodes in a line. 
Additionally, they can show that oracles are not needed if the problem is transformed into a non-decision variant. 
In~\cite{KoutsopoulosSS15} the idea is generalized to a protocol framework that solves the \fdp without being reliant on a certain topology and is thereby combinable with most existing overlay protocols.

\subsection{Our contribution}
To the best our knowledge, this paper presents the first attempt to have stricter requirements towards the self-stabilization process in topological self-stabilization.
We define and study \emph{monotonic searchability}, which captures a typical use case for overlay networks, i.e., searching other nodes.
More formally, we want to guarantee for a self-stabilizing topology that once a search message from node $u$ to another node $v$ is successfully delivered, all future search messages from $u$ to $v$ succeed as well.
We focus on studying non-trivial monotonic searchability for the list topology.
First, we show that in general it is impossible to provide non-trivial monotonic searchability from any initial system state, due to the presence of certain initial messages.
This justifies to study searchability only for so-called \emph{admissible system states} in which  these messages are not present anymore, as long as the protocol gurantees convergence to these states.
We give a self-stabilizing list protocol and an appropriate search protocol that achieve the desired goal and prove their correctness.
Moreover, we broaden the elaborateness of the problem statement, by allowing nodes to leave the line topology, i.e., solving the Finite Departure Problem in addition to the aforementioned problems. 
Also for this combination of problems we present suitable protocols and prove their correctness. 

\section{Preliminaries}
\label{sec:preliminaries}

Since gone nodes will never execute any action, we only consider initial states in which all nodes are present. 
We also restrict the initial state to contain only a finite number of messages that can trigger actions specified by our protocol, since other messages are ignored by the nodes. 
Finally, we do not allow the presence of references that do not belong to a node in the system.
From now on, an initial system state satisfies all of these constraints.

The following propositions are restatements of results in~\cite{corona} and imply further necessary conditions on initial system states.
\begin{enumerate}
\item If a compare-store-send program solves the linearization problem, each computation starts in a weakly connected initial state.
\item If a compare-store-send program solves the linearization problem, each computation starts in a state in which all references belong to present nodes.
\end{enumerate}

A \emph{message invariant} is a predicate of the following form:
If there is a message $m$ in the incoming channel of a node, then a predicate $P'$ must hold.
A protocol may specify one or more message invariants.
An arbitrary message $m$ in a system is called \emph{corrupted} if the existence of $m$ violates one of the message invariants.
A state $s$ is called \emph{admissible} if there are no corrupted messages in $s$.
We say a protocol \emph{admissible-message satisfies} a property if the following two conditions hold: (i) in computations in which every state is admissible, it satisfies the property, and (ii) starting from any initial state, there is a computation suffix in which every state is admissible.
A protocol \emph{unconditionally satisfies} a property if it satisfies this property starting from any state.

With this notion in mind, we can show that admissible-message satisfaction is necessary for non-trivial monotonic searchability for any routing algorithm $R$.

\begin{lemma}\label{lem:admissible_message_necessary_for_monotonic_searchability}
If a compare-store-send self-stabilizing protocol satisfies non-trivial monotonic searchability then this protocol must be admissible-message satisfying.
\end{lemma}
The structure of the proof is as follows: we consider an arbitrary unconditionally satisfying protocol and show that it does not satisfy monotonic searchability by creating a bad instance for this protocol.
In particular, we exploit that our model does not ensure FIFO delivery of messages.

\begin{proof}
  Assume there is a compare-store-send self-stabilizing protocol that unconditionally satisfies non-trivial monotonic searchability.
  First of all, note that if it violates only the second condition of admissible-message satisfiability, then there are computations in which monotonic searchability is never satisfied, implying that it cannot satisfy non-trivial monotonic searchability.
  Thus, assume that the first condition is violated, i.e., the protocol satisfies the property in computations with arbitrary message, regardless of any invariants.
  Consider the network given in Figure~\ref{fig:no_corrupt}.
  \begin{figure}[h]
    \centering
 	\begin{tikzpicture}
		\node(u)[circle,draw=black] at (0,0) {$u$};
		\node(v)[circle,draw=black] at (1,0) {$v$};
		\node(w)[circle,draw=black] at (2,0) {$w$};
		\draw (u) -- (v);
		\draw[dashed,->] (v) to[out=50,in=-230] (w);
	\end{tikzpicture}
	\caption{Instance for this proof.}\label{fig:no_corrupt}
 \end{figure}
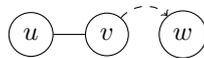
 
 The implicit edge $(v,w)$ is in $v.Ch$.
 We carry out the proof as a game between the protocol and an adversary: based on the decisions of the protocol,  the adversary may decide on the delivery speed of messages, and imitate additional messages at each node.
 The latter is possible since nodes can not distinguish between these messages and messages from an initial state that have not been received yet.
 Furthermore, the adversary may set the initial state of the nodes.
 
 At first, we issue a $search(u,w)$ request in $u$ that we denote by $a$ in the following.
 We argue that the adversary can force $u$ to forward $a$ to $v$.
 Therefore note the following:
 \begin{enumerate}
  \item As long as $u$ does not receive any further messages, $u$ does not know any other node, so $v$ is the only possible next hop for $a$.
  \item If $u$ tries to wait for a certain amount of time before sending $a$, the adversary simply halts the system for that time, i.e., no messages are delivered in that timeframe and the system state stays the same.
  \item If $u$ requires the receipt of another message in order to forward $a$, the adversary imitates this message at $u$.
  \item If $u$ relies on its internal state to forward $a$, the adversary changes the initial state of $u$ such that it does not forward any message, which contradicts the assumption that non-trivial monotonic searchability is satisfied.
  Therefore, $u$ must not rely on its state to forward $a$.
  \item There are no other conditions that $u$ can wait on.
 \end{enumerate}
 Therefore, $u$ will send out $a$ to $v$ eventually. At the point in time when $u$ does so, we issue a second $search(u,w)$ request in $u$.
 For similar reasons as stated above, $b$ must be sent to $v$ at some point in time as well.
 
 Since both messages are in $v.Ch$ and the adversary is allowed to decide message speeds, it lets $v$ receive $b$ first.
 Node $v$ has no explicit edge to $u$ and the adversary can enforce that the implicit edge $(v,w)$ will not be received by $v$ until $v$ handles $b$. Therefore $b$ must be answered with ``FAIL'' at some point in time (since the $b$ cannot be forwarded anymore) and $u$ will be informed about that.

 Next, the adversary causes the edge $(v,w)$ to arrive at $v$.
 Since the protocol must stabilize to the line, at some point in time, the edge $(v,w)$ will be established.
 Until then, the adversary withholds message $a$ in $v.Ch$.
 Afterwards, when $a$ arrives at $v$, it can be forwarded to $w$ and thus correctly served.
 
 Therefore, message $a$ succeeds, whereas message $b$ that was sent after message $a$ fails.
 This is a contradiction to the assumption that the protocol achieves non-trivial monotonic searchability. 
\end{proof}

Consequently, to prove non-trivial monotonic searchability for a protocol (according to a given routing protocol $R$) it is sufficient to show that: (i) the protocol has a computation suffix in which every state is admissible and (ii) the protocol guarantees non-trivial monotonic searchability according to $R$ in admissible states.

For the \fdp, it was shown in~\cite{departure1}, there is no distributed protocol within our model that can decide when it is safe for a node $u$ to leave the system and thereby solve the \fdp.
The authors circumvent this impossibility result with the help of oracles.
In general, an \emph{oracle} is a predicate that depends on the current system state and the node calling it. 
In the context of the \fdp, an oracle is supposed to advise a leaving node when it is safe to execute \textbf{exit}.
We use the oracle \nidec as introduced in~\cite{departure1} in order to solve the \fdp. 
\nidec evaluates to \textbf{true} for a node $u$ calling it, if no node $v \neq u$ has a reference to $u$ in its local memory or in a message in $v.Ch$ and if $u.Ch$ is empty. For an in depth discussion of oracles for the \fdp, we refer the reader to~\cite{departure1,KoutsopoulosSS15}.

\section{The \blp and the \srp protocols}
\label{sec:blpAlgorithm}

In this section, we present the \blp protocol and the \srp protocol. \blp solves the linearization problem and is admissible-message satisfying non-trivial monotonic searchability according to \srp.
Note that any protocol satisfying non-trivial monotonic searchability must be admissible-message satisfying as shown in Section~\ref{sec:preliminaries}.
This section is organized as follows: 
First, we describe \blp and \srp in detail (Subsection~\ref{subsec:blp_desription}).
Then, we prove that the \blp protocol solves the linearization problem (Subsection~\ref{sec:self_stabilization_proof}).
Last, we prove that the \blp protocol satisfies non-trivial monotonic searchability according to \srp (Subsection~\ref{sec:monotonic_searchability_proof}).
From now on we drop the "according to \srpwithoutxspace" clause, since we only consider searchability for \srp.

\subsection{Description of \blp and \srp}
\label{subsec:blp_desription}
The \blp Protocol builds upon the protocol introduced in~\cite{self-stabilizing-list} that solves the linearization problem.
For this protocol, every node only keeps a single left and right neighbor. 
If a node $u$ receives a reference of a node $v$ with $u<v$ ($u>v$, respectively), $u$ either saves $v$ as its new right (left) neighbor if $v$ is closer to $u$ than the current right (left) neighbor $w$ and delegates the reference of $w$ to $v$ or (in case $v$ is not closer), $v$ is not saved and delegated to $w$.
Here, \emph{delegation} means that the reference of a node is sent in a message to another node and not kept in the local memory.
A natural (local) search protocol for this topology is to always forward search requests to the neighbor closest to the desired target node, or to abort the search request in case no such neighbor exists.
Note that these easy and elegant protocols cannot guarantee monotonic searchability due to three simple facts: (i) due to delegation, it is possible that an explicit edge $(u,v)$ is replaced by an explicit edge $(u,w)$ and an implicit edge $(w,v)$, (ii) consequently, $u,v$ are not in the same weakly connected component in $ENG$ (even though they were before delegation) and (iii) searchability is defined for $ENG$.

The \blp protocol introduces the following changes in order to satisfy monotonic searchability:
Instead of having a single left and right neighbor, a node $u$ has sets of neighbors $Left$ and $Right$ (that it sorts implicitly according to id).
In the following, whenever we use the notation $Left(u)$/$Right(u)$, we refer to these sets of a node $u$.
The main principle that we use is that every node $w$ does not delegate any edge to a node $v$ stored in $Left(w)$ or $Right(w)$ directly.
Instead it first introduces (using \introduce{v,w}) this node to another node $u$, waits for an acknowledgement that the edge has been added to $Left(u)$ or $Right(u)$ (which is basically the \linearize{v} message) and then delegates the edge to a node closer to $v$ (using \tempdelegate{v}).
More specifically, whenever a node $u$ has multiple neighbors to one side, it does not delegate edges to the closest neighbor directly, but does the following.
W.l.o.g. assume that it has multiple neighbors $w_1,\ldots,w_\ell$ to the right with $id(w_i)<id(w_{i+1})$.
In the \timeout action $u$ introduces $w_i$ to $w_{i-1}$, with an \introduce{w_i,u} message.
Thereby, $w_{i-1}$ knows that it got the reference from $u$, saves the reference to $w_i$ directly, sends a \linearize{w_i} message back to $u$ and a \tempdelegate{u} to itself (the latter is only to preserve connectivity).
Node $u$ can now react to that \linearize{w_i} message, by deleting $w_i$ from its memory and sending the reference to the closest node to the left of $w_i$ in $Right$ (which is not necessarily $w_{i-1}$ anymore). 
Thereby, $u$ preserves a path of explicit edges between $u$ and $w_i$.
Additionally, $u$ sends its own reference to the closest neighbors with a $\introduce{u,\bot}$ message who turn this into a \tempdelegate{u} message.
In general, the \tempdelegate{u} action is used to delegate an implicit edge to a node $u$ into one direction (i.e., to the left or to the right) as long as there is a node between the current node and $u$ in $Left$ or $Right$.
Note that implicit edges are not used for search, thus we do not have to apply the principle of introducing first and delegating afterwards for this kind of edges.
However, we have to delegate in order to preserve connectivity and to stabilize to the line eventually.
Note that, even though a node has temporarily more references than necessary for the final line topology our protocol still eventually stabilizes to the line, as we will show later.
The pseudocode for all \blp actions is given in Listing~\ref{algo:blp}.
Note that a node refers to itself with the expression $self$.
Additionally, keep in mind that the timeout action is the only action that is not triggered as a result of another action.
Instead, is triggered regularly.

\begin{lstlisting}[mathescape=true,float=*,caption=\blp protocol,label=algo:blp]
$\timeout$
 for all $destID \in Waiting$
   send $forwardProbe(self, destID, \{self\}, self.seq)$ to $self$
 //Let $Left = \{v_1, v_2, \dots, v_k\}$ with $id(v_1) < id(v_2) < \dots < id(v_k)$ 
 for all $v_i \in Left$ with $1 \leq i < k$
   send $\introduce{v_i,self}$ to $v_{i+1}$
 //Let $Right = \{w_1, w_2, \dots, w_l\}$ with $id(w_1) < id(w_2) < \dots < id(w_l)$    
 for all $w_i \in Right$ with $1 < i \leq l$
   send $\introduce{w_i,self}$ to $w_{i-1}$
 send $\introduce{self,\bot}$ to $v_1$
 send $\introduce{self,\bot}$ to $w_1$

$\introduce{v,w}$ 
 if($id(v) < id(self)$)
   if{$w \neq \bot$}
     $Left \gets Left \cup \{v\}$
     send $\linearize{v}$ to $w$
     send $\tempdelegate{w}$ to $self$
   else //$w = \bot$
     send $\tempdelegate{v}$ to $self$
 else if($id(v) > id(self)$) 
   //Analogous to the previous case.
    
$\linearize{v}$
 send $\tempdelegate{v}$ to $self$
 if($id(v) < id(self)$)
   if($Left \neq \emptyset$)
     $x \gets argmax \{id(x') | x' \in Left \}$
     if($v \neq x$)
       $w \gets argmin \{id(w') | w' \in Left\ und\ id(w')>id(v)\}$
       $Left \gets Left \setminus \{v\}$
       send $\tempdelegate{v}$ to $w$
 else if($id(v) > id(self)$) 
   //Analogous to the previous case. 
 
$\tempdelegate{u}$
 if($id(u) < id(self)$)
   if($Left = \emptyset$)
     $Left \gets Left \cup \{u\}$ 
   else //$Left \neq \emptyset$
     $x \gets argmax \{id(x') | x' \in Left \}$
     if($id(x) < id(u)$)
       $Left \gets Left \cup \{u\}$
     else if($id(x) > id(u)$)
       send $\tempdelegate{u}$ to $x$
 else if{$id(u) > id(self)$} 
   //Analogous to the previous case. 

\end{lstlisting}

The \srp protocol works as follows:
Whenever the \initsearch{destID} action is called at a node $u$, $u$ creates a new \search{u,destID} message and starts to periodically initiate \forwardprobe{u,destID, \{u\}, self.seq} messages that it sends to itself. 
In the following, assume $id(u)<destID$ (the other case is analogous).
Each \forwardprobe{} message has a set of nodes, called $Next$ attached to it, which contains the nodes the message will visit in its future. 
It also has a counter $seq$ attached to it whose meaning we will explain later.
Whenever a \forwardprobe{u,destID, Next, seq} message is at a node $w$, $w$ removes itself from $Next$ and adds all its right neighbors $x$ with $id(x) \leq destID$ to $Next$. 
Then it forwards the \forwardprobe{u,destID, Next, seq} message to the node with minimal id in $Next$.
If a \forwardprobe{u,destID, Next, seq} message arrives at a node $v$ with $id(v)=destID$, it directly responds with a \psuccess{destID,seq, v} message to $u$.
However, if $Next$ is empty at a node $w$ with $id(w) \neq destID$ after $w$ has added the aforementioned right neighbors, the \forwardprobe{} message is answered with a \pfail{destID,seq} message.
In any case, as soon as $u$ receives the response, it acts accordingly: If the answer to a \forwardprobe{u,destID,Next,seq} message is a \pfail{destID,seq} message, it drops the corresponding \search{u,destID} message completely.
If the answer is \psuccess{destID,v}, \search{u,destID} messages waiting at $u$ are directly sent to $v$.

Note that if additional \search{u,destID} messages are created at $u$ while $u$ is still waiting for an answer to an earlier initiated \forwardprobe{u,destID}, these requests simply wait together with the previous request (realized by simple $WaitingFor[destID]$ field) and are aborted or sent as soon as the \pfail{destID} or \psuccess{destID,v} response arrives at $u$, (i.e., search requests to the same destination are sent out in batches if possible).
Furthermore, note that nodes do not memorize whether they have already sent \forwardprobe{} messages to a certain destination. Due to corrupt initial states, this knowledge could be wrong and nodes relying on this knowledge would wait forever.
Therefore, nodes periodically send \forwardprobe{} messages, instead of only once.
Note that because we make no assumptions on the message delivery speed and channels are not FIFO, it is possible that \pfail{} messages arrive at a node $u$ that are answers to \forwardprobe{} messages initiated long ago.
However, in the meantime, there might have been successful responses.
To deal with this, each node $u$ stores a sequence number counter $seq$.
Whenever \initsearch{destID} is executed by $u$ and there is no \search{u,destID} that waits for an answer to a \forwardprobe{u,destID, Next, seq} message, $u$ increments $u.seq$, stores the new $u.seq$ value in an entry for $v$ and always attaches the current sequence number ($u.seq$) to each \forwardprobe{} message $u$ sends.
Responses to probes (success and failure) sent by $u$ also contain this sequence number.
Whenever a response is sent back to $u$, $u$ checks whether the sequence number in this message is at least the sequence number stored for $destID$.
If not, it simply drops the message, since in that case, the answer belongs to a \forwardprobe{} message sent for an earlier batch of \search{u,destID} messages that have already been processed.
The complete pseudocode for \srp is given in Listing~\ref{algo:search}.

\begin{lstlisting}[mathescape=true,float,caption=\srp protocol,label=algo:search]
$\initsearch{destID}$
 create new message $m=\search{self,destID}$
 if($WaitingFor[destID]=\emptyset$)
   $WaitingFor[destID] \gets \{\}$
   $self.seq \gets self.seq + 1$
   $seq[destID] \gets self.seq$
 //Store the messages to $WaitingFor$ 
 $WaitingFor[destID] \gets WaitingFor[destID] \cup \{m\}$ 

$\forwardprobe{source,destID,Next,seq}$
 if($destID = id(self)$)
   if($Next \neq \emptyset$)
     for all $u \in Next$
       send $\tempdelegate{u}$ to $self$
   send $\psuccess{destID, seq, self}$ to $source$
   send $\tempdelegate{source}$ to $self$
 else //$destID \neq id(self)$
   if($destID > id(self)$)
     $Next \gets Next\setminus \{self\} \cup \{ w \in Right | id(w) \leq destID\}$
     if($Next = \emptyset$)
       send $\pfail{destID, seq}$ to $source$
       send $\tempdelegate{source}$ to $self$
     else //$Next \neq \emptyset$
       $u \gets argmin\{ id(u) | u \in Next\}$
       if($id(u) < id(self)$)
         send $\tempdelegate{u}$ to $self$
       else if($id(u) < id(argmin\{ id(v) | v \in Right\})$)
         $Right \gets Right \cup \{u\}$
       send $\forwardprobe{source,destID,Next,seq}$ to $u$
   else if($destID < id(self)$) 
     //Analogous to the previous case.
 
$\psuccess{destID,seq,dest}$
 if($seq \geq seq[destID]$)
   /* The message belongs to currently  
    * stored search requests to $dest$. */
   send all $m \in WaitingFor[destID]$ to $dest$
   $WaitingFor[destID] \gets \emptyset$
 send $\tempdelegate{dest}$ to $self$
 
$\pfail{destID,seq}$
 if($seq \geq seq[destID]$)
   /* The message belongs to currently  
    * stored search requests to $dest$. */
   $WaitingFor[destID] \gets \emptyset$
\end{lstlisting}

In order to not unnecessarily blow up the pseudocode, we intentionally left out a sanity check for each node, i.e., before executing each action, each node $u$ makes sure that $Left$ only contains nodes $v$ with $v < u$ and that $Right$ only contains nodes $v$ with $u < v$.
If this is not the case for some node $v$, $u$ rearranges the reference to $v$ accordingly.
This way, in every computation, the following lemma holds:

\begin{lemma}\label{lem:left_and_right_are_what_they_say}
  For every node $v$ it holds: For all $x \in Left$, $id(x) < id(v)$, and for all $y \in Right$, $id(v)<id(y)$.
\end{lemma}

\subsection{\blp solves the linearization problem}\label{sec:self_stabilization_proof}
In this section, we prove the following theorem:
\begin{theorem}\label{thm:blp_solves_linearization}
 \blp is a self-stabilizing solution to the linearization problem. 
\end{theorem}
We prove the theorem in three steps: 
First, we show that starting from any initial state in which $NG$ is weakly connected, $NG$ will always be weakly connected.
Second, we show that starting from any initial state, there will be a state in which $ENG$ will be a supergraph of the line graph and that the explicit edges corresponding to the line will never be removed.
Third, we prove that all superfluous explicit edges will eventually vanish.

The first step is represented by the following lemma:
\begin{lemma}\label{lem:NG_remains_weakly_connected}
 If a computation of \blp starts from a state where $NG$ is weakly connected then in every state, $NG$ remains weakly connected.
\end{lemma}

\begin{proof}
  First, note that in every action whenever a message with a reference to a node $v$ is received by a node $u$ then either $v$ is added to the set $Left(u)$ or $Right(u)$ or a new message is created with $v$ as a parameter and sent to a node $w \in Left(u) \cup \{u\} \cup Right(u)$.
  Thus, the implicit edge $(u,v)$ is replaced by a path $(u,w,v)$.
  
  Furthermore, the only action for that removes a reference to $v$ from one of the sets $Left(u)$ or $Right(u)$ is the \linearize{v} action.
  However, in \linearize{v}, if $v$ is removed from $Left(u)$ or $Right(u)$, $v$ is also introduced to a node $w$ in $Left(u)$ or $Right(u)$.
  Thus, the edge $(u,v)$ is replaced by a path $(u,w,v)$ in this case, too.
\end{proof}

For the second step of the proof of the theorem, we introduce the notation $nextLeft(u) := argmax \{id(v) | v \in Left(u)\}$ and $nextRight(u) := argmin \{id(v) | v \in Right(u)\}$.
Furthermore, let $length(u,v)$ for two nodes $u$ and $v$ denote the hop distance in the (ideal) line topology between $u$ and $v$.
We define $rv(v)$ for a node $v$ as $length(v,nextRight(v))$ if $Right(v) \neq \emptyset$ or as $n$ if $Right(v) = \emptyset$; we define $lv(v)$ analogously for $nextLeft(v)$.
With this, we define a potential function $\Phi := \sum_{i=1}^{n-1}rv(v_i) + \sum_{i=2}^{n}lv(v_i)$ where $v_1 < v_2 < \dots < v_n$ are all nodes ordered by their id increasingly.
Notice that $\Phi$ is bounded from above by $2n(n-1)$ and from below by $2(n-1)$.
Also notice that according to the protocol, $nextLeft(v)$ ($nextRight(v)$) can only change if $v$ puts a node closer to $v$ than $nextLeft(v)$ ($nextRight(v)$) into $Left$ ($Right$).
Thus, $\Phi$ never increases.
We define the \emph{closest neighbor graph} as the graph $G_{NB}=(V,E_{NB})$ where $V$ is the set of all nodes and $(x,y) \in E_{NB}$ iff $y=nextRight(x) \lor y=nextLeft(x)$.
Furthermore, we say an edge is \emph{temporary} if it is an implicit edge due to a \tempdelegate{} message.
 All other types of implicit edges are called \emph{non-temporary}.
One can show the following:
\begin{lemma}\label{lem:closest_neighbor_graph_bidirected_and_strongly_connected}
 Assume there is a system state such that $\Phi$ does not decrease in any further step of the computation.
 Then $G_{NB}$ is bidirected and strongly connected.
\end{lemma}
We prove this lemma step-by-step, starting with the following lemma:
\begin{lemma}\label{lem:closest_neighbor_graph_bidirected}
 Assume a system state such that $\Phi$ does not decrease in any further step of the computation.
 Then $G_{NB}$ is bidirected.
\end{lemma}

\begin{proof}
 Assume for contradiction there exists an edge $(x,y) \in E_{NB}$ such that $(y,x) \notin E_{NB}$ and w.l.o.g. assume $x < y$.
 This implies $nextRight(x)=y$ and $x\neq nextLeft(y)$.
 Since $\Phi$ does not change any more, $y$ will remain $nextRight(x)$ and eventually by the fair action execution assumption, \timeout will be executed in $x$ and $x$ will send an \introduce{x,\bot} to $y$, which, by the fair message receipt assumption, will be eventually delivered to $y$.
 This implicit edge will turn into a temporary edge $(y,x)$.
 Note that if $Left(y) = \emptyset$ or $nextLeft(y) < x$, then, according to the protocol and because $x < y$, $nextLeft(y)$ will be replaced by $x$ causing $\Phi$ to decrease, which contradicts to the initial assumption.
 Therefore, $Left(y) \neq \emptyset$ and $x < nextLeft(y) < y$ must hold.
 According to the protocol, $(y,x)$ will be delegated (first to $nextLeft(y)$, then possibly further) until it reaches at a node $z$ with $z = \emptyset$ $nextLeft(z) < x < z$.
 Here similar arguments as above yield a contradiction. 
 Thus, $G_{NB}$ must be bidirected.
\end{proof}
The definition of a closest neighbor graph and Lemma~\ref{lem:left_and_right_are_what_they_say} imply the following:
\begin{corollary}\label{cor:only_topology_is_line}
 If $G_{NB}$ is bidirected and disconnected, every connected component forms a line.
\end{corollary}
To show that $G_{NB}$ is also strongly connected, we need two additional lemmata.
We start with the following:
 \begin{lemma}\label{lem:non_temporary_will_become_expl_or_temporary}
  Assume that in a state of the computation of \blp $G_{NB}$ is bidirected and disconnected.
  If there is a non-temporary edge $(w,v)$ with $w \in C_1, v \notin C_1$ for a connected component $C_1$, then eventually either there will be an explicit or a temporary edge $(x,y)$ with $x \in C_1$ and $y \notin C_1$ or $\Phi$ will decrease.
 \end{lemma}
 
 \begin{proof}
W.l.o.g., assume $w<v$.
First of all, note that according to the protocol, if the graph $G_{NB}$ changes, $\Phi$ must decrease.
Since in that case we are done, in the following we assume that $G_{NB}$ will never change.
Furthermore, by Corollary~\ref{cor:only_topology_is_line}, the connected components of $G_{NB}$ form a line.
We now make a case distinction over all possible types of $(w,v)$:
 \begin{enumerate}
  \item $(w,v)$ is an implicit edge from a \forwardprobe{m} message in which $v = source$ or $v \in Next$ and $id(w)=destID$. 
	Then once the message is received, $(w,v)$ will be turned into a temporary edge and the claim follows.
   \item $(w,v)$ is an implicit edge from a \forwardprobe{m} message in which $v=source$ and $destID > id(w)$.    
	Consider the state in which this message is received and the corresponding action is executed.
	Then $Next$ is updated according to the protocol.
	If $Next$ is empty after this operation, a temporary edge $(w,v)$ is established and the claim holds.	
	Otherwise, let $u:= argmin\{id(u)|u \in Next\}$ after the update.
	Note that if $u > w$, we have two sub-cases: Either $minRight(w) > u$ or $minRight(w) \leq u$.
	In the former case, $u$ will be added to $Right(w)$, causing $\Phi$ to decrease, and the claim holds.
	In the latter case, due to the way $Next$ was updated, $minRight(w) = u$ must hold.
	Applying the previous arguments recursively yields that the message will, at some point in time, reach at a node $x \in C_1$ where either $destID = id(x)$ or $Next = \emptyset$ after the update.
	In this case, a temporary edge $(x,v)$ will be established.
	\\
	Now, consider the case $u < w$.
	Again, we have two sub-cases: Either $u \notin C_1$ or $u \in C_1$.
	In the former case, since the protocol establishes the temporary edge $(w,u)$, the claim follows.
	In the latter case, the message will be forwarded to $u \in C_1$.
	According to the protocol, for $u' := argmin\{id(u')|u' \in Next\}$ after the update of $Next$, it holds $u' > u$. Thus, this case reduces to the other case above.
  \item $(w,v)$ is an implicit edge from a \forwardprobe{m} message in which $v=source$ and $destID < id(w)$.    
	This case is analogous to the previous one.
  \item $(w,v)$ is an implicit edge from a \forwardprobe{m} message in which $v\in Next$ and $destID > id(w)$.    	
	Note that in this case if a \forwardprobe{m} message is delegated from a node $x$ to a node $y < x$, then a temporary edge $(x,y)$ is also established.
	Then either $y \notin C_1$ directly proving the claim, or $y \in C_1$.
	Observe that each \forwardprobe{m} message can only be delegated from a node $x$ to a node $y < x$ once.
	Thus, either starting from the first or the second step, whenever a \forwardprobe{m} message is delegated from a node $x$ to a node $y$, then $y > x$.
	Furthermore, note that the protocol assures $y \in Right(x)$, i.e., $y \in C_1$ as well.
	The only case when a \forwardprobe{m} message is no longer delegated is if $Next$ is empty (in which there is nothing left to prove), or when $destID = id(x)$ for a node $x$.
	In the latter case, for each node remaining in $Next$, a temporary edge is created.
  \item $(w,v)$ is an implicit edge from a \forwardprobe{m} message in which $v\in Next$ and $destID > id(w)$.
	This case is analogous to the previous one.
  \item $(w,v)$ is an implicit edge from a \psuccess{} (in which $v$ is $Dest$) message and a temporary edge $(w,v)$ will be established.
  \item $(w,v)$ is an implicit edge from an \introduce{} message.
	Note that according to the protocol, all edges in an \introduce{} message are added either as explicit edges or as temporary edges.
  \item $(w,v)$ is an implicit edge from a \linearize{} message and $(w,v)$ will be turned into a temporary edge.
 \end{enumerate} 
\end{proof}
 \begin{lemma}\label{lem:temporary_will_shorten_or_phi_will_decrease}
  Assume that in a state of the computation of \blp $G_{NB}$ is bidirected and disconnected.
  If there is an explicit or a temporary edge $(w,v)$ with $w \in C_1$ and $v \notin C_1$ for a connected component $C_1$, then eventually there will be an explicit or temporary edge $(x,y)$ with $x \in C_1, y \notin C_1$ and $length(x,y)<length(w,v)$, or $\Phi$ will decrease.
 \end{lemma}

\begin{proof}
  W.l.o.g., assume $w<v$.
  First, assume $(w,v)$ is an explicit edge. 
  If $v = nextRight(w)$, we have a contradiction to the assumption $w \in C_1$ and $v \notin C_1$.
  Thus $w < nextRight(w) < v$ must hold. 
  In this case, in \timeout a new edge $(x,v)$ with $w < x < v$ will be introduced and the claim will hold.
  Second, assume that $(w,v)$ is an implicit edge from a \tempdelegate{} message. 
  Then either $v < nextRight(w)$ and $(w,v)$ turns into an explicit edge and $v$ becomes $nextRight(w)$, causing $\Phi$ to decrease, or a \tempdelegate{v} message is sent to $nextRight(w)$ resulting in a shorter edge $(nextRight(w),v)$.
  This completes the proof of the second claim.
\end{proof}

We are now ready to prove \textbf{Lemma~\ref{lem:closest_neighbor_graph_bidirected_and_strongly_connected}}:
\begin{proof}
  Assume there is an initial state in which $\Phi$ does not decrease anymore.
  Furthermore, assume that the closest neighbor graph $G_{NB}$ is disconnected.
  Firstly, Lemma~\ref{lem:closest_neighbor_graph_bidirected} guarantees that $G_{NB}$ is bidirected.
  Furthermore, by Lemma~\ref{lem:NG_remains_weakly_connected}, there must be at least one (implicit or explicit) edge $(w,v)$ between a connected component $C_1$ and another connected component.
  Together with Lemma~\ref{lem:non_temporary_will_become_expl_or_temporary} this implies that at some point there must be a temporary or explicit edge $(x,y)$ with $x \in C_1$ and $y \notin C_1$.
  However, then Lemma~\ref{lem:temporary_will_shorten_or_phi_will_decrease} can be applied.
  Since there is only a finite number of times that there can be a shorter edge, at some state, $\Phi$ must decrease, yielding a contradiction.
  Thus $G_{NB}$ must be weakly connected.
  Note that Lemma~\ref{lem:closest_neighbor_graph_bidirected} implies that $G_{NB}$ is also strongly connected, yielding the claim of Lemma~\ref{lem:closest_neighbor_graph_bidirected_and_strongly_connected}.
\end{proof}
Note that since $\Phi$ can never increase and since $\Phi$ is bounded from below, $\Phi$ can only decrease for a finite number of states.
After that, the conditions of Lemma~\ref{lem:closest_neighbor_graph_bidirected_and_strongly_connected} are fulfilled.
This lemma and Corollary~\ref{cor:only_topology_is_line} imply the following corollary:
\begin{corollary}\label{cor:explicit_edge_graph_will_become_supergraph}
 For any computation of \blp, there is a state in which the graph formed by the explicit edges is a supergraph of the line topology.
\end{corollary}
For the third step of the proof of the theorem, we have the following lemma:
\begin{lemma}\label{lem:superfluous_edges_will_vanish}         
 If a computation of \blp contains a state in which $ENG$ is a supergraph of the line topology, then there will be a suffix in which $ENG$ is the line topology and no new explicit edges will ever be created again.
\end{lemma}

\begin{proof}
    For the proof, we introduce the following notatation: We say an implicit edge $(u,v)$ is \emph{right-relevant} if $u < v$ and the implicit edge $(u,v)$ is due to a $\introduce{v,w}$ message in $u.Ch$ for $w \neq \bot$.
    Accordingliy, we say an edge $(u,v)$ is \emph{left-relevant} if $v < u$ and the implicig edge $(u,v)$ is due to a $\introduce{v,w}$ message in $u.Ch$ for $w \neq \bot$.
    Additionally, we call an explicit edge $(u,v)$ \emph{superfluous} if $v \neq nextRight(u) \land v \neq nextLeft(u)$.

  Consider the state in which the graph formed by the explicit edges is a supergraph of the line topology.
  First of all, notice that according to the protocol, an explicit edge  that belongs to the line topology will never be removed (because this would require a node $u$ to get acquainted with a node $v$ that is closer than $minLeft(u)$ or $minRight(u)$ which is not possible).
  In addition, notice that according to the protocol, in every state (right-/left-)relevant edges are the only implicit edges that can be turned into an explicit edge any more.
  Notice that a right-relevant edge $(u,v)$ can only be created by a node $w < u$ with a superfluous explicit edge to $v$.
  Thus, for every node $u$ it holds: if there is no node $w < u$ with a relevant or superfluous edge $(w,u)$, then there will never be a relevant or superfluous edge $(x,u)$ with $x < u$ again.
  
  Consider the leftmost node $u$ that either has at least one right-relevant edge or at least one superfluous right neighbor.
  Note that once all right-relevant edges have been received by $u$, then no node $x \leq u$ will ever add a superfluous right neighbor again.
  Furthermore, notice that right-relevant edges will be turned into explicit edges upon receipt.
  Now, for every superfluous right neighbor $v$ of $u$, $u$ will send an \introduce{v,u} to some node $w \in Right(u)$.
  Each of these will eventually be received and, according to the protocol, be answered with a \linearize{v} message at $u$.
  This will cause $u$ to delegate $v$ to a node $x > u$.
  After the last superfluous edge has been delegated, no node $x \leq u$ will ever have a superfluous right neighbor again.
  
  Continuing this approach, we can show that all superfluous right neighbors will eventually vanish.
  Using analogous arguments, we can also show that all superfluous left neighbors will eventually vanish.
  Thus, the lemma follows.
\end{proof}
Note that Corollary~\ref{cor:explicit_edge_graph_will_become_supergraph} and Lemma~\ref{lem:superfluous_edges_will_vanish} imply that \blp converges to the list.
Moreover, Lemma~\ref{lem:superfluous_edges_will_vanish} yields the closure property.
This finishes the proof of Theorem~\ref{thm:blp_solves_linearization}.

\subsection{\blp satisfies non-trivial monotonic searchability}\label{sec:monotonic_searchability_proof}
In this subsection we prove the following theorem:
\begin{theorem}\label{thm:blp_guarantees_monotonic_searchability}
 \blp admissible-message satisfies non-trivial monotonic searchability according to \srp.
\end{theorem}

We start with some preliminaries.
First we define $R(v)$ as the set of all nodes $x$ with $id(v) < id(x)$ for which there is a directed path from $v$ to $x$ consisting solely of explicit edges $(y,z)$ with $id(y) < id(z)$.
Furthermore, we define $R(v,ID) := \{x \in R(v) | id(x) \leq ID\}$.
In addition, we define $L(v)$ as the set of all nodes $x$ with $id(x) < id(v)$ for which there is a directed path from $v$ to $x$ consisting solely of explicit edges $(y,z)$ with $id(z) < id(y)$.
For a set $U$, $R(U) := U \cup \bigcup_{u \in U}R(u)$ and $R(U,ID) := \{x \in R(U) | id(x) \leq ID\}$.
Accordingly, $L(U) := U \cup \bigcup_{u \in U}L(u)$ and $L(U,ID) := \{x \in L(U) | id(x) \geq ID\}$.

Moreover, we define a state as admissible if the following message invariants hold:
\begin{enumerate}
 \item If there is an \introduce{v,w} message with $w \neq \bot$ in $u.Ch$, then $v \neq w$, and $u \in R(w)$ (or $u \in L(w)$).
 \item If there is a \linearize{v} message in $w.Ch$, then there is a node $u \neq v$ with $u \in Right(w)$ and $v \in R(u)$ if $w < v$ (or $u \in Left(w)$ and $v \in L(u)$ if $v < w$).
 \item If there is a \forwardprobe{source,destID,Next,seq} message in $u.Ch$, then
    \begin{enumerate}
      \item $id(source) < destID$ and $\forall x \in Next: id(x) \geq id(u)$ and $u = argmin_u\{id(u) | u \in Next\}$ 
      (alternatively $destID < id(source)$ and $\forall x \in Next: id(x) \leq id(u)$ and $u = argmax_u\{id(u) | u \in Next\}$).
      \item $id(source) < destID$ and $R(next) \subseteq R(source)$ (or $destID < id(source)$ and $u \in L(source)$).
      \item if $v$ exists such that $id(v) = destID$ and $id(source) < destID$ and $v \notin R(Next,destID)$ (or $id(source) < destID$ and $v \notin L(Next,destID)$) then for every admissible state with $source.seq[destID] < seq$, $v \notin R(source,destID)$ ($v \notin L(source,destID)$).
    \end{enumerate}
 \item If there is a \psuccess{destID, seq, dest} message in $u.Ch$, then $id(dest) = destID$ and $dest \in R(u)$ if $destID > id(u)$ (or $dest \in L(u)$ if $destID < id(u)$).
 \item If there is a \pfail{destID, seq} message in $u.Ch$, then either there is no node with id $destID$, or for every admissible state with $u.seq[destID] < seq$, $v \notin R(u)$ (and $v \notin L(u)$), where $v$ such that $id(v)= destID$.
 \item If there is a \search{v, destID} message in $u.Ch$, then $id(u) = destID$ and $u \in R(v)$ if $id(v) < destID$ (or $u \in L(v)$ if $destID < id(v)$).
\end{enumerate}

\begin{lemma}\label{lem:once_admissible_always_admissible}
 If in a computation of \blp, there is an admissible state, then all subsequent states are admissible.
\end{lemma}

In order to prove Lemma~\ref{lem:once_admissible_always_admissible}, we need the following lemmata:
\begin{lemma}\label{lem:once_first_two_invariants_hold_then_always}
 If in a computation of \blp, the first two invariants hold, then in all subsequent states the first two invariants will hold.
\end{lemma}
\begin{proof}
 Assume there is a state $s_1$ in which the first two invariants hold and such that in the (direct) subsequent state $s_2$ one of the first two invariants does not hold.
 Obviously, this can only be due to one of the following three reasons:
 First, a new \introduce{v,w} message with $w \neq \bot$ was sent to a node $u$ with $u \notin R(w)$ (and $u \notin L(w)$) in $s_1$.
 Second, a new \linearize{v} message was sent to a node $w$ in $s_1$, but there is no node $u \neq v$ with $u \in Right(w)$ and $v \in R(u)$ (or $u \in Left(w)$ and $v \in L(u)$).
 Third, a node $y$ was removed from a set $Right(w)$ (or $Left(w)$).
 We show that all three cases cannot happen.
 
 For the first case, notice that according to the protocol, the only occasion when an \introduce{v,w} message with $w \neq \bot$ is sent is in the \timeout action of a node $w$.
 Here, it is only sent to nodes in $Right(w)$ (or $Left(w)$) and only with a first parameter $v \neq w$.
 
 For the second case, notice that according to the protocol, the only occasion when a \linearize{v} message is sent to a node $w$ is in an \introduce{v,w} action at a node $u'$.
 This must have been triggered by an \introduce{v,w} message with $w \neq \bot$.
 Thus, before the action was executed, by the first invariant, $u' \in R(w)$ (or $u' \in L(w)$) and $v \neq w$ were both fulfilled.
 This implies that there must be a node $u \in Right(w)$, i.e., $w < u$ such that $u' \in R(u)$ or $u' = u$ (or a node $u \in Left(w)$, i.e., $u < w$, such that $u' \in L(u)$ or $u' = u$).
 During the execution of the action, $v$ was added to $Right(u')$ (or $Left(u')$), which implies $v \in R(u)$ (or $v \in L(u)$).
 
 For the third case, note that a node $y$ is only removed from $Right(w)$ (or $Left(w)$) if the \linearize{y} action has been executed in $w$ between $s_1$ and $s_2$.
 However, by the second invariant, there must be a node $u \neq y$ with $u \in Right(w)$ and $y \in R(u)$ (or $u \in Left(w)$ and $y in L(u)$).
 Thus, after the removal, $y \in R(w)$ still holds.
 
 Therefore, in all three cases the first two invariants cannot be violated and have to hold in $s_2$, too.
\end{proof}

\begin{lemma}\label{lem:R_grows_monotonically}
  If there is a state in which the first two invariants hold, and $x \in R(v)$ ($x \in L(v)$), then in every subsequent step, $x \in R(v)$ ($x \in L(v)$).
\end{lemma}
\begin{proof}
We only consider the case $x \in R(v)$, as $x \in L(v)$ is completely analogous.

Obviously, adding additional edges does not remove elements from $R(v)$.
The only action that delegates away an explicit edge $(y,z)$ stored in $Right(y)$ for some nodes $y,z$ (and hence could remove nodes from $R(v)$) is the \linearize{} action if $y < z$.
Therefore, consider an arbitrary \linearize{z} action executed by $y$.
Note that since we assumed that the first two invariants hold, right before \linearize{z} is executed, it has to hold that there is a node $u \neq z$ with $u \in Right(y)$ and $z \in R(u)$, by the second invariant.
Consequently, after $z$ is removed from $Right(y)$, $z \in R(y)$ still holds.
\end{proof}

\begin{lemma}\label{lem:once_first_three_invariants_hold_then_always}
 If in a computation of \blp, the first three invariants hold, then in all subsequent states the first three invariants will hold.
\end{lemma}
\begin{proof}
  Assume there is a state $s_1$ in which the first three invariants hold and such that in the (direct) subsequent state $s_2$ one of the first three invariants does not hold.
  Note that by Lemma~\ref{lem:once_first_two_invariants_hold_then_always} the first two invariants cannot be violated in $s_2$.
  Furthermore, by Lemma~\ref{lem:R_grows_monotonically} and the fact that $u.seq[id]$ is monotonically increasing (according to the protocol), one can easily show that the only reason why Invariant 3 can be invalidated is that a new \forwardprobe{} message is sent.
  In the following, we will only consider the case, where $id(source) < destID$, as the other case is completely analogous.
  
  	Assume a node $x$ sends a \forwardprobe{source,destID,Next,seq} message to a node $y$.
	This may happen in two cases: Either in the \timeout action of a node $x$, or when $x$ receives another \forwardprobe{source,destID,Next',seq} message and executes the corresponding action.
	In the first case, $Next = \{x\}$ and it is easy to see that claim a) and b) of the third invariant are fulfilled.
	In the second case, both $\forall z \in Next': id(z) \geq id(y)$ and $y = argmin_u\{id(u) | u \in Next'\}$ hold, since (by the third invariant) (i) $\forall z \in Next': id(z) \geq id(x)$, and $\forall z \in Right(x): id(z) \geq id(x)$ (by Lemma~\ref{lem:left_and_right_are_what_they_say}), (ii) only nodes from $Right(x)$ are added to $Next$, (iii) $x$ was $argmin_u\{id(u) | u \in Next\}$ and is not added to $Next'$, and (iv) $y$ is selected as the minimum node from $Next'$. 
	By the third invariant, $x \in R(source)$, which implies $Right(x) \subseteq R(source)$.
	Now, since $R(Next') \subseteq R(source)$ by the third invariant and $Next = Next' \setminus \{x\} \cup Right(x)$, $R(Next) \subseteq R(source)$.
	Thus Invariant~3b) holds afterwards.
	
	For the third claim of the third invariant, we again distinguish between the message being sent in \timeout or in the \forwardprobe{source,dest,Next',seq} action.
	In the former case, notice that $R(Next,destID) = R(source,destID)$.
	Assume there has been an admissible state in which $source.seq[destID] < seq$ and $v \in R(source,destID)$ hold.
	Since $source.seq[destID]$ is monotonically increasing, this must have been a previous state.
	By Lemma~\ref{lem:R_grows_monotonically}, $v \in R(source,destID) = R(Next,destID)$ must still hold, yielding a contradiction.
	In the latter case, assume $v \in R(Next',destID)$ (otherwise, Invariant~3c) trivially holds).
	Notice that due to Invariant~3b), $x \in R(source)$.
	Since the only node that is in $R(Next',destID)$ but not in $R(Next,destID)$ is $x$, $v\in R(Next,destID)$ follows.
	
  Thus, the first three invariants still hold in $s_2$.
\end{proof}

\begin{lemma}\label{lem:once_first_five_invariants_hold_then_always}
 If in a computation of \blp, the first five invariants hold, then in all subsequent states the first five invariants will hold.
\end{lemma}
\begin{proof}
  Assume there is a state $s_1$ in which the first five invariants hold and such that in the (direct) subsequent state $s_2$ one of the first five invariants does not hold.
  Note that by Lemma~\ref{lem:once_first_three_invariants_hold_then_always} none of the first three invariants can be violated in $s_2$.
	Furthermore, by Lemma~\ref{lem:R_grows_monotonically} and the fact that according to the protocol, $u.seq[id]$ is monotonically increasing, one can check that the only reason for why Invariant 4, or 5 can be invalidated is that a new \psuccess{}, or \pfail{} message is sent.
	In the following, we will only consider the case, $id(u) < destID$, as the other cases are completely analogous.

	First, we consider \psuccess{} messages. Hence, assume that a node $x$ sends a \psuccess{destID,seq,dest} message to a node $u$.
	According to the protocol, this may only be in a \forwardprobe{} action, when a \forwardprobe{source,destID,Next,seq} message has arrived at $x$ with $id(x) = destID$ and $u = source$.
	By b) of the third invariant, $dest \in R(u)$.
	
	For the \pfail{} messages, assume a node $x$ sends a \pfail{destID, seq} message to a node $u$.
	According to the protocol, this may only be in a \forwardprobe{} action, when a \forwardprobe{source,destID,Next,seq} message has arrived at $x$ with $id(x) \neq destID$, $u = source$ and $Next = \{x\}$ and there is no $y$ in $Right(x)$ with $id(y) \leq destID$.
	If no node with id $destID$ exists, we are done.
	Otherwise, we have that $v \notin R(Next,w)$.
	By c) of the third invariant, this implies the claim.
  
	Therefore, the first five invariants have to hold in $s_2$, too.
\end{proof}
Using these lemmata, we can prove \textbf{Lemma~\ref{lem:once_admissible_always_admissible}}:
\begin{proof}
 Assume there is an admissible state $s_1$ such that in the (direct) subsequent state $s_2$ is not admissible.
 Let $s_1$ be the first such state.
	Note that by Lemma~\ref{lem:once_first_five_invariants_hold_then_always}, none of the first five invariants can be violated in $s_2$.
	Furthermore, by Lemma~\ref{lem:R_grows_monotonically} one can check that the only reason for why Invariant~6 can be invalidated is that a new \search{} message, is sent.
	In the following, we will only consider the case, $id(u) < destID$, as the other case is completely analogous.

	Assume a node $x$ sends a \search{v, destID} message to a node $u$.
	According to the protocol, $x = v$, and $v$ must have received a \psuccess{destID,seq,u}, for which, by Invariant~4, $id(u) = destID$, and $u \in R(v)$ must hold, i.e., the sixth invariant holds.
	
	Therefore, all invariants have to hold in $s_2$, too.
\end{proof}

\begin{lemma}\label{lem:admissible_state_always_exists}
 In every computation of \blp there is an admissible state.
\end{lemma}

\begin{proof}
According to Theorem~\ref{thm:blp_solves_linearization}, there is a state $s_1$ in which and in every subsequent state, every node $x$ has at most one node in $Right(x)$ and at most one nide in $Left(x)$.
Note that according to the protocol, any \introduce{v,w} message with $v \neq w$ is only sent from a node $w$ with more than one in $Right(w)$ or $Left(x)$.
Thus, by the fair message receipt assumption, there will be a state $s_2$ after $s_1$, in which all such messages have been received.
Further note that any \linearize{v} message is only sent from a node $u$ if $u$ received an \introduce{v,w} message, which cannot be the case in $s_2$.
Thus, by the fair message receipt assumption, there will be a state $s_3$ after $s_2$, in which all \linearize{} message have been received.
This implies that the first two invariants hold in $s_3$.
By Lemma~\ref{lem:once_first_two_invariants_hold_then_always}, they will do so in every subsequent state.

Next we show that starting from $s_3$, every \forwardprobe{source,destID,Next,seq} violating the third invariant will have vanished at some point in time.
In the following we only consider such messages with $id(source) < destID$ (the other case is analogous).
First, notice that any \forwardprobe{} message initiated in a \timeout action by a node $x$ cannot violate the third invariant.
This is obvious for a) and b).
For c), notice that if $v$ with $id(v) = destID$ exists and $v \notin R(Next,w)$ and there is an admissible state with $x.seq[destID] < seq$ and $v \in R(x)$, then according to the protocol this state must have been an earlier state and Lemma~\ref{lem:R_grows_monotonically} implies that $v \in R(x)$ in the current state, yielding a contradiction.

Second, note that any existing \forwardprobe{} message $m$ can cause at most one other \forwardprobe{} message $m'$ to be created when it is received by a node $x$.
If this $m$ does not violate the third invariant then since the first two invariants hold, $m'$ will also not violate the third invariant (for reasons similar to those in the proof of Lemma~\ref{lem:once_first_three_invariants_hold_then_always}).
Thus, we will show that every \forwardprobe{} message that violates the third invariant can only cause a finite number of \forwardprobe{} messages that violate the third invariant (which will eventually be received and thus disappear).
First of all, note that every \forwardprobe{} message $m$ violating Invariant~3a) cannot cause a \forwardprobe{} message $m'$ violating Invariant~3a) according to the protocol.
Thus, after all initial \forwardprobe{} messages have been received, Invariant~3a) holds for every \forwardprobe{} message.
Now, observe that any such \forwardprobe{} message which is received by a node $x$ can only initiate a new \forwardprobe{} message to a node $y$ with $id(y) > id(x)$, according to the protocol.
Since there is only a finite number of nodes, this implies that all \forwardprobe{} message violating Invariant~3 will eventually disappear.

Now, consider the state $s_4$ in which all of the first three invariants hold.
Note that by Lemma~\ref{lem:once_first_three_invariants_hold_then_always}, they hold for all subsequent states, too.
Notice that any \psuccess{} or \pfail{} message in $u.Ch$ for a node $u$ cannot cause $u$ to send a \psuccess{} or \pfail{} message.
The only only action in which a new \psuccess{} or \pfail{} message is sent is in the \forwardprobe{} action of a node.
Such an action requires the receipt of a \forwardprobe{source,destID,Next,seq} message $m$ for which, by definition of $s_4$, the third invariant holds.
Note that according to the protocol $m$ can only cause a \psuccess{destID,seq,dest} message $m'$ that is sent to sent to a node $x$, if $id(u) = destID$ (i.e., $dest = u$) and $x = source$.
By Invariant~3b), $u \in R(source)$, implying $dest \in R(x)$, i.e., the fourth invariant holds regarding $m'$.
A \pfail{destID,seq,dest} message $m'$ to a node $x$ can only be caused by $m$ if $id(u) < destID$ and $Next \setminus \{u\} \cup \{w \in Right| id(w) \leq destID\} = \emptyset$, implying that $v \notin R(Next,destID)$ for a node $v$ with $id(v) = destID$.
By Invariant~3c), for every admissible state with $source.seq[destID] < seq$, $v \notin R(source,destID)$, i.e., the fifth invariant holds regarding $m'$.
All in all, there is a state $s_5$ such that all \psuccess{} or \pfail{} messages that were in the incoming channel of any node in $s_4$ have been received and consequently, for all \psuccess{} and \pfail{} messages the fourth and fifth invariant will hold.
By Lemma~\ref{lem:once_first_five_invariants_hold_then_always}, they hold for all subsequent states, too.

Consider this state $s_5$.
Notice that \search{v,destID} message can only be sent to a node $u$ from a \psuccess{destID,seq,u} action in $v$, which requires the receipt of a \psuccess{destID,seq,u} message for which, by definition of $s_5$, the fourth invariant holds.
This implies, $destID = id(u)$ and $u \in R(v)$, yielding Invariant~6 for the new message.
Thus, in the state $s_6$ after all \search{} messages that were in the incoming channel of any node in $s_5$ have been received, all invariants hold, i.e., $s_6$ is an admissible state.
\end{proof}

Lemma~\ref{lem:once_admissible_always_admissible} and Lemma~\ref{lem:admissible_state_always_exists} imply the following Corollary~\ref{cor:admissible_suffix}.

\begin{corollary}
\label{cor:admissible_suffix}
    In every computation of \blp, there exists a suffix in which every state is admissible. 
\end{corollary}
For the rest of this subsection, we assume that every computation starts in an admissible state, since we want to show monotonic searchability must hold starting from admissible states only.
Furthermore, w.l.o.g., we only consider \search{u,destID} messages with $id(u) < destID$.

Before we can prove Theorem~\ref{thm:blp_guarantees_monotonic_searchability}, we need an additional result:
\begin{lemma}\label{lem:w_in_R_leads_to_message_success}
 For every message $m = \forwardprobe{v,destID, Next, seq} \in u.Ch$ with $id(u) < destID$, it holds that if there is a node $w$ with $id(w)=destID$ and $w \in R(u)$, then there will be a state with $m' = \forwardprobe{v,destID, Next', seq} \in w.Ch$.
\end{lemma}
In order to prove Lemma~\ref{lem:w_in_R_leads_to_message_success}, we need the following additional lemma:
\begin{lemma}\label{lem:forwardprobe}
Assume for a \forwardprobe{v,destID, Next, seq} message $m \in x.Ch$, there is a $u \in R(Next,destID)$.
Then either $u = x$ or there will be a state in which a \forwardprobe{v,destID, Next', seq} message is  in $y.Ch$ for some node $y$ with $id(y) > id(x)$ and $u \in R(Next',destID)$.
\end{lemma}
\begin{proof}
	Note that when $m$ is received by $x$, a new message with $Next' = Next\setminus\{x\} \cup Right(x)$ will be sent.
	According to the third invariant, for all nodes $z$ in $Next$, $id(z) \geq id(x)$ holds, and $x$ is the node with minimum id among all nodes in $Next$.
	By Lemma~\ref{lem:left_and_right_are_what_they_say}, the same holds for the nodes $z$ in $Right(x)$.
	Thus, $x$ is the node with minimum id among all ones in $R(Next,w)$ and for the node $y$ to which a new \forwardprobe{v,destID,Next',seq} message is sent it holds that $id(y) > id(x)$.
	Furthermore, $R(Next(x),destID)\setminus \{x\} \subseteq R(Next',destID)$.
	Thus, also $u \in R(Next',destID)$ and the claim follows.
\end{proof}
Using this, we can prove \textbf{Lemma~\ref{lem:w_in_R_leads_to_message_success}}:
\begin{proof}
	Note that when $m$ arrives as $u$, $Next$ will be changed such that $R(Next,w) = R(u,w)$.
	If $w \in R(u)$, then $w \in R(Next,w)$ afterwards.
	Thus, by applying Lemma~\ref{lem:forwardprobe} recursively, we have that eventually  a \forwardprobe{v,destID, Next', seq} is in $w.Ch$, which will be received according to the fair message receipt assumption.	
\end{proof}

We are now ready to prove Theorem~\ref{thm:blp_guarantees_monotonic_searchability}:
\begin{proof}
  Let $m, m'$ be two \search{u,destID} messages initiated in $u$ in admissible states with $m$ being initiated before $m'$ and assume that $m$ is delivered successfully, but $m'$ is not.
	Let $v$ be such that $id(v) = destID$.
	Note that if $m'$ is added to the set $WaitingFor[destID]$ when $m$ is already in the set, then the protocol will handle both messages identical, i.e., if $m$ is successfully delivered to $v$ due to an \psuccess{} message, $m'$ is as well.
	Therefore, $m'$ is added to $WaitingFor[destID]$ when $m \notin WaitingFor[destID]$, which implies $u.seq[destID]$ has increased since the successful delivery of $m$ (according to the protocol).
	Since we assume that $m'$ is not delivered successfully, either a \pfail{dest,seq} message eventually arrives at $u$ with $seq \geq u.s[destID]$, or no \psuccess{destID,seq,dest} with $seq \geq u.s[destID]$, $dest = destID$ will ever arrive at $u$. 
	We consider both cases individually.
	In the first case, by the fifth invariant, $v \notin R(u)$ has to hold even though $m$ was already successfully delivered.
	By the sixth invariant, when $m$ was delivered, $v \in R(u)$, which is why this is a contradiction to Lemma~\ref{lem:R_grows_monotonically}.
	In the second case, note that \forwardprobe{u,destID,\{u\},seq} messages are regularly initiated by $u$ with $seq \geq u.s[destID]$ (since $u.seq$ is monotonically increasing).
	Again, due to the successful delivery of $m$, by the sixth invariant and Lemma~\ref{lem:R_grows_monotonically}, $v \in R(u)$ when $m'$ was initiated, and therefore, by Lemma~\ref{lem:w_in_R_leads_to_message_success}, a \forwardprobe{u, destID, Next', seq} message with $seq \geq u.s[destID]$ will eventually be in $v.Ch$, which will be answered with a \psuccess{destID, seq, v} message, causing $m'$ to be sent to $v$.	
	By the fair message receipt assumption, this contradicts the assumption that $m'$ is not successfully delivered.
\end{proof}

\section{The \blpp and the \srpp protocols}
\label{sec:theBLPPAlgorithm}
For the \blp protocol in Section~\ref{sec:blpAlgorithm} we implicitly assumed a static node set, i.e., nodes are not allowed to leave or join the network. In this section we want investigate monotonic searchability in terms of the \emph{Finite Departure Problem} (\fdp) of~\cite{departure1}.
Naturally, a leaving node does not execute \initsearch{}, since it aims at leaving the system. Additionally, a leaving node that is the destination of a \forwardprobe{} message, will deliberately answer with \pfail{}.
Consequently, monotonic searchability can only be maintained for pairs of staying nodes.

We note that the \fdp deliberately ignores that new nodes can join the network. 
However, this abstraction is justified in a self-stabilizing setting, since from an algorithmic point of view for some node $u$ a new node joining the network is the same as getting a message from a node that it has never been in contact with.

In this section, we present the \blpp and the \srpp protocols.
In the following sections, we further show that \blpp solves the \fdp (Section~\ref{sec:FDP_solution_proof}) and also the linearization problem (Section~\ref{sec:blpp_linearization_proof}), and extend the proofs of Section~\ref{sec:monotonic_searchability_proof} to show that \blpp also satisfies non-trivial monotonic searchability according to \srpp (Section~\ref{sec:blpp_monotonic_searchability_proof}).

\subsection{Description of \blpp and \srpp}

For two staying nodes that interact with each other, \blpp is analogous to \blp.
Therefore, we only specify the changes in case a node itself is leaving or receives a message from a leaving node.
A leaving node distinguishes between two different kinds of neighbors: those that it already had before switching to the leaving mode (which are $Left$ and $Right$ from \blp) and those which it received while being leaving (\templeft and \tempright). 
Searchability is only preserved for nodes in the former two sets.

For the \forwardprobe{}, \introduce{}, \linearize{} and \tempdelegate{} actions, a leaving node $u$ will always save nodes in \templeft and \tempright in cases where a staying node saves them in $Left$ and $Right$. 
In its \timeout action, a leaving node $u$ either introduces all its neighbors to each other and executes \textbf{exit} if \nidec is true or it sends a \revandlinREQ{} message to all neighbors.
With this \revandlinREQ{dir} message $u$ requests all neighbors to stop holding its reference. 
As it was shown in~\cite{departure1}, leaving nodes should never send their own reference for a successful departure protocol. 
Therefore, a \revandlinREQ{dir} message only contains a value $dir \in \{left,right\}$ that indicates whether a left or right neighbor should be removed, i.e., $u$ sends a \revandlinREQ{left} message to all its neighbors to the right and and a \revandlinREQ{right} message to all its neighbors to the left.
If a  node $v$ receives a \revandlinREQ{dir} message, there are two possible scenarios. 
If $v$ is staying, it sends a \revandlinACK{v,uniqueValue} message to all neighbors in the given direction, which contains its own reference and for each neighbor a uniquely created value (i.e., in our case a local counter or the $id$ of a node would be sufficient).
This values is also saved as satellite data by $v$ at the corresponding node reference in the neighbor set.
If $v$ is leaving, it behaves like a staying node if the $dir$ is right; otherwise it ignores the request. 
Thereby, leaving nodes with a higher id are given a higher priority for exiting the system.
Once a leaving node $u$ receives a \revandlinACK{v,uniqueValue} message, it responds with \revandlin{nodeList, uniqueValue} message that contains the received unique value (for identification purposes) and also all its neighbors that are on the opposite of the node in the message (i.e., if the received node is to the right of $u$, $u$ sends all left neighbors and vice-versa). 
A \revandlinACK{v,uniqueValue} message is ignored by a staying node, meaning that it is transformed into a \tempdelegate{v} to itself.
Finally, the \revandlin{nodeList, uniqueValue} message is received by $v$ and $v$ checks if it has a neighbor with the given unique value. 
If this is the case, $v$ either finishes the reversal process by deleting the reference to $u$ and saving the newly received neighbors (if $v$ is staying or getting the \revandlin{nodeList, uniqueValue} message from a right neighbor) or $v$ ignores the message by simply saving all nodes in \templeft (if $v$ is leaving and  getting the \revandlin{nodeList, uniqueValue} message from a left neighbor).  
In case the unique value does not match, the \revandlin{nodeList, uniqueValue} message is not a response to a former \revandlinACK{v,uniqueValue} message and all received nodes are processed by \tempdelegate{} messages to $v$ itself.

The \srpp protocol is very similar to the \srp protocol. 
As already mentioned, leaving nodes will neither execute \initsearch{}, nor will they send out a \psuccess{} message. 
In fact the only action that is different in multiple places is the \forwardprobe{} action, since we have to make sure that references are not saved in $Left$ and $Right$ but in \templeft and \tempright.

Similar to \blp, \blpp performs a sanity check for \templeft, \tempright, $Left$ and $Right$ before each action. 
The same is done for the $nodeList$ received in a \revandlin{} message. 
However, in the last case a failing sanity check (i.e., the nodes in $nodeList$ are from two different sides of the current node) directly implies that the message is corrupt and it is safe to process the nodes with \tempdelegate{}.
The pseudocode for \blpp and \srpp is presented in Algorithms~\ref{algo:fdp} and~\ref{algo:search2}.

\begin{lstlisting}[mathescape=true,caption=\blpp protocol,label=algo:fdp]
$\timeout$
 if($self.mode = staying$) // See Algorithm$~\ref{algo:blp}$.
 else
   if($ \mathcal{NIDEC}$)
     for all $v \in Left \cup Right \cup \templeft \cup \tempright$ 
       for all $w \in Left \cup Right \cup \templeft \cup \tempright$
         send $v.\introduce{w,\bot}$ to $v$
         send $\introduce{v,\bot}$ to $w$
     $\textbf{exit}$
   else    
     for all $v \in Left \cup \templeft$ 
       send $\revandlinREQ{right}$ to $v$
     for all $w \in Right \cup \tempright$ 
       send $\revandlinREQ{left}$ to $w$

$\introduce{v,w}$
 if($id(v) < id(self)$)
   if($self.mode = staying$) // See Algorithm$~\ref{algo:blp}$.
   else       
     if($v \notin Left$} 
       $\templeft \gets \templeft \cup \{v\}$   
     if{$w \neq \bot \land w \notin Left$} 
       $\templeft \gets \templeft \cup \{w\}$ 
 else if($id(v) > id(self)$) //Analogous to the previous case.

$\linearize{v}$
 if($id(v) < id(self)$)
   if($self.mode = staying$) 
     // See Algorithm$~\ref{algo:blp}$.
   else 	
     $\templeft \gets \templeft \cup \{v\}$ 
 else if($id(v) > id(self)$) //Analogous to the previous case.
   
$\tempdelegate{u}$
 if($id(u) < id(self)$)
   if($Left = \emptyset$)
     if($self.mode = staying$) 
       $Left \gets Left \cup \{u\}$
     else 
       $\templeft \gets \templeft \cup \{u\}$
   else  
     $x \gets argmax \{id(x') | x' \in Left \}$
     if($id(x) < id(u)$)
       if($self.mode = staying$) 
         $Left \gets Left \cup \{u\}$
       else 
         $\templeft \gets \templeft \cup \{u\}$
     else 
       send $\tempdelegate{u}$ to $x$
 else if($id(u) > id(self)$) //Analogous to the previous case.   
\end{lstlisting}
\begin{lstlisting}[mathescape=true,caption=\blpp protocol (continued)] 
$\revandlinREQ{dir}$
 if{$dir = right$}
   for all $v \in Right \cup \tempright$ 
     if($uniqueValues[v] = \bot$) // i.e., $v$ does not exist in uniqueValues.
       /* Assume that generateUniqueValue() creates a unique value.
       $uniqueValues[v] = self.generateUniqueValue()$
     send $\revandlinACK{self, uniqueValues[v]}$ to $v$
 else if($dir = left \land self.mode = staying$) 
   // Analogous to the previous case.  

$\revandlinACK{v,uniqueValue}$
 if($id(v) < id(self)$)
   if($self.mode = leaving$)
     $\templeft \gets \templeft \cup \{v\}$
     send $\revandlin{Right, uniqueValue}$ to $v$
   else
     send $\tempdelegate{v}$ to $self$
 else if ($id(v) > id(self)$) 
   // Analogous to the previous case. 

$\revandlin{nodeList,uniqueValue}$
 if($\exists v \in Left \cup Right \cup \templeft \cup \tempright$ with $uniqueValues[v] = uniqueValue$)
   if($self.mode = staying$)
     if($id(v) < id(self)$)
       $Left \gets Left \cup nodeList$ 
       $Left \gets Left \setminus \{v\}$
       send $\introduce{self,\bot}$ to $v$
     else if($id(v) > id(self)$) 
       //Analogous to the previous case.
   else //$self.mode = leaving$
     if($id(v) < id(self)$)
       $\templeft \gets \templeft \cup nodeList$ 
     else //$id(v) > id(self)$
       if($v \in Right$)
         $Right \gets Right \cup nodeList$
         $Right \gets Right \setminus \{v\}$
       else
         $\tempright \gets \tempright \cup nodeList$
         $\tempright \gets \tempright \setminus \{v\}$
       send $\introduce{self,\bot}$ to $v$
 else
   for all $u \in nodeList$
     send $\tempdelegate{u}$ to $self$
\end{lstlisting}

\newpage
\begin{lstlisting}[mathescape=true,caption=\srpp protocol,label=algo:search2]
$\initsearch{destID}$
 if($self.mode = staying$) 
   //See Algorithm~$\ref{algo:search}$.
 else
   // do nothing.  

$\forwardprobe{source,destID,Next,seq}$
 if($destID = id(self)$)
   if($self.mode = staying$) 
     //See Algorithm~$\ref{algo:search}$.
   else
     send $\pfail{destID,seq}$ to $source$
     for all $u \in Next$ 
       send $\tempdelegate{u}$ to $self$
     send $\tempdelegate{source}$ to $self$
 else
   if($destID > id(self)$)
     $Next \gets Next\setminus \{self\} \cup \{ w \in Right | id(w) \leq destID\}$
     if($Next = \emptyset$)
       send $\pfail{destID,seq}$ to $source$
       send $\tempdelegate{source}$ to $self$
     else
       $u \gets argmin\{ id(u) | u \in Next\}$
       if($id(u) < id(self)$)
         send $\tempdelegate{u}$ to $self$
       else if ($id(u) < id(argmin\{ id(v) | v \in Right\})$)
         if{$self.mode = staying$}
           $Right \gets Right \cup \{u\}$
         else
           $\tempright \gets \tempright \cup \{u\}$
       send $\forwardprobe{source,destID,Next,seq}$ to $u$
   if($destID < id(self)$) 
     //Analogous to the previous case.

$\psuccess{destID,seq,dest}$
 if($self.mode = staying$) 
   //See Algorithm$~\ref{algo:search}$.
 else
   send $\tempdelegate{dest}$ to $self$

$\pfail{destID,seq}$
 if($self.mode = staying$) 
   // See Algorithm$~\ref{algo:search}$.
\end{lstlisting}

In the following sections we will show that (i) \blpp is a self-stabilizing solution to the \fdp, (ii) \blpp is a self-stabilizing solution to the linearization problem and (iii) \blpp admissible-message satisfies non-trivial monotonic searchability according to \srpp.

\subsection{\blpp solves the \fdp}
\label{sec:FDP_solution_proof}
This section is dedicated to prove the following theorem.

\begin{theorem}\label{thm:blpp_solves_fdp}
 \blpp is a self-stabilizing solution to the \fdp.
\end{theorem}

First of all, we prove the \emph{safety} property. Let $PNG$ be the subgraph of $NG$, whose nodes are all present nodes. 

\begin{lemma}
\label{lem:fdp:safety}
 If a computation of \blpp starts in a state in which $PNG$ is weakly connected, $PNG$ remains weakly connected in every state of this computation.
\end{lemma}
\begin{proof}
 Note that the result of Lemma~\ref{lem:NG_remains_weakly_connected} still holds for the actions \timeout, \introduce{}, \linearize{} and \tempdelegate{} in case the executing node is staying.
 Furthermore, the result directly transfers to \introduce{}, \linearize{} and \tempdelegate{} if the executing node is leaving, since the only change is that references are stored in \templeft and \tempright instead of $Left$ and $Right$.
 The same is true for the actions of \srpp: \forwardprobe{}, \psuccess{} and \pfail{}.
Moreover, a leaving node executing the \timeout action can only endanger weak connectivity, if it executes \textbf{exit}.
However, in that situation \nidec is true for the node and it introduces all neighbors to each other before calling the \textbf{exit} command.
Hence, weak connectivity is also nevertheless preserved for all present nodes.

For the three new actions of \blpp we note that the only action that actively deletes a reference is \revandlin{}. 
However, if that happens, an \introduce{} message containing the own reference is sent to the deleted node.
Thus an explicit edge $(a,b)$ is replaced by an implicit edge $(b,a)$ (i.e., the edge is \emph{reversed}) and weak connectivity is preserved.
\end{proof}
Second, we prove the \emph{Liveness} property:
\begin{lemma}\label{lem:fdp:liveness}
 For any computation of \blpp there exists a computation suffix in which all leaving nodes are gone.
\end{lemma}
 \begin{proof}
   Assume for contradiction there is a computation $C$ of \blpp for which there does not exist a computation suffix in which all leaving nodes are gone.
   Let $CS_1$ be the suffix of $C$ in which (i) all nodes that will ever decide to be leaving have done so and (ii) all leaving nodes that will execute \textbf{exit} are gone.
Since the node set is finite such a suffix has to exist.
   Let $s_1$ be the first state of $CS_1$.

  Let $CS_2$ be the suffix in which all \introduce{}, \linearize{}, \tempdelegate{}, \revandlinREQ{}, \revandlinACK{}, \revandlin{}, \psuccess{}, \pfail{} and \forwardprobe{}, messages that were in the incoming channel of any node in state $s_1$ have been received and all \revandlin{} messages sent in response to a \revandlinACK{} in $s_1$ have also been received.
  Note that for all states in $CS_2$ it holds that holds that $dest$ is staying, since leaving nodes answer every \forwardprobe{} with a \pfail{}.
    Additionally, leaving nodes do not send \forwardprobe{} messages in $CS_1$ so the number of  \forwardprobe{} message in $CS_2$ for which the $source$ is leaving is upper bounded.
  In fact, for $CS_2$ it holds that any \forwardprobe{} message has been received at least once.
  Therefore, a node cannot be added twice to the $Next$ field of a message since \forwardprobe{}messages are only forwarded into one direction according to the protocol, i.e., a \forwardprobe{} will visit only nodes with increasing id or only with decreasing ids.
  Therefore, each \forwardprobe{} can only be forwarded finitely many often and is thereby answered by \psuccess{} or \pfail{} eventually.
  Consequently, there is also a state (and thereby a computation suffix $CS_3$), in which all \forwardprobe{} message which have a leaving node as the $source$ are answered by their \psuccess{} or \pfail{} and also these \psuccess{} or \pfail{} messages in the incoming channel of a leaving node have been received.

  Note that in every state of $CS_3$, every message that is in $x.Ch$  has been sent in $CS_1$.
  We call the node that adds a message into the incoming channel the \emph{sender} of the message.
  By the definition of $CS_3$, the following invariants hold (which is easy to check, according to the protocol):
  \begin{enumerate}
   \item If \forwardprobe{source,destID,Next,seq} message is in $x.Ch$ and $id(source) < destID$, then for all $y \in Next$ with $y \neq x$: $id(y) > id(x)$ and when the sender $z$ sent the \forwardprobe{source,destID,Next,seq} to $x$, either $x \in Right(z) \cup \tempright(z)$ or $z = x$.
   \item If there is a \tempdelegate{y} message in $x.Ch$ and $id(x) < id(y)$, then for the sender $z$, $id(z) < id(x) < id(y)$ or $z = x$.
   \item If there is an \introduce{y,z} message in $x.Ch$ with $z \neq bot$ and $id(x) < id(y)$, then $z$ is the sender and when $z$ sent the \introduce{y,z} message, $y \in Right(w)$ (and vice-versa).
   \item If there is an \introduce{y,\bot} message in $x.Ch$ then either $y$ is also the sender and $y$ is not leaving (since otherwise $y$ would execute \textbf{exit} after sending the message contradicting the definition of $C$) or the sender $z \neq y$ is staying and sent the message as an answer to an \linearize{y} message.
   \item If there is a \linearize{y} message in $x.Ch$ and $id(x) < id(y)$, then in the state in which the sender $z$ sent the \linearize{y} message, it must have done so in response to an \introduce{y,x} it received.
   \item If there is a \revandlinACK{y,uniqueValue} message in $x.Ch$, then $y$ is the sender, and $id(y) < id(x)$ and in the state in which $y$ sent the message, $x \in Right(y) \cup \tempright(y)$ (or $id(x) < id(y)$ and $y$ is staying).
   \item If there is a \revandlin{nodeList,uniqueValue} in $x.Ch$ then the sender $z$ must be leaving and for every $y \in NodeList$, it holds that $id(y) > id(z)$. Additionally, the message is a response due to a \revandlinACK{x,uniqueValue} message received by $z$.
  \end{enumerate}
In order to prove the desired statement, we first show two additional lemmas before continuing with the proof.

\begin{lemma}
\label{lem:staying_node_stops_pointing}
Consider a state $s_3$ of $CS_3$ and let $u$ be a staying node and $v$ be a leaving node with $id(u)<id(v)$. 
If it holds in $s_3$ that (i) there is no edge $(u',v) \in NG$ with $id(u')<id(u)$, 
and (ii) for any leaving node $v'$ with $id(u)<id(v')<id(v)$ there will never be an edge $(u,v') \in NG$ in a subsequent state, then there is a state $s'$ in $CS_3$ such that for the computation suffix $CS'$ starting in $s'$ it holds that $(u,v) \notin NG$ for every state in $CS'$.
\end{lemma}

\begin{proof}
Since there is no edge $(u',v) \in NG$ with $id(u')<id(u)$, no node to the left of $u$ can add a message to $u.Ch$ that contains the reference of $v$.
Additionally, since $id(u)<id(v)$ no node $x$ to the right of $u$ can add a \tempdelegate{v} or \introduce{v,x} message to $u.Ch$, according to the protocol.
Moreover, no leaving node to the right can add a message to $u.Ch$ that contains the reference of $v$ (i.e., a \revandlin{} message).
This is due to the fact that a \revandlin{} message to sent $u$ by a leaving node $v'$ with $id(u)<id(v')<id(v)$ can only be sent as a response to a \revandlinACK{u,uniqueValue} by $v'$ (see Invariant~7), which cannot happen since there will never be an edge $(u,v') \in NG$.
Note that we only consider states in $CS_3$, therefore the above mentioned invariants hold.

At first assume that no edge $(u,v)$ exists.
If never gets a reference to $v$ in $CS_3$ the lemma holds trivially.
Consequently, $u$ can only get the reference of $v$ in an \introduce{v,\bot} or in a \linearize{v}  message.
In the first case, the \introduce{v,\bot} was sent by a node $w\neq v$ as a response to a former \linearize{v} message, according to Invariant~4.
According to the pseudocode of \linearize{}, this can only happen if $id(v)>id(u)>id(w)$ or $id(v)<id(u)<id(w)$.
Both cases cannot happen since $id(u)<id(v)$ and no node to the left of $u$ can add a message to $u.Ch$.
So in this scenario, the lemma holds as well.
In the second case, the \linearize{v} message $u$ will send a \tempdelegate{v} it itself.
Consequently, there is a state in $CS_3$ in which an edge $(u,v)$ exists, which is handled in the following.

Now consider the case that an edge $(u,v)$ exists.
Note that $(u,v)$ can be a multi-edge and be explicit as well as implicit. 
In fact, it can be both and if it is implicit it can be due to multiple messages in $u.Ch$.
At first we show that all messages in $u.Ch$ that contain a reference to $v$, will be made explicit or vanish completely.
\begin{itemize}
\item If there is an \introduce{v,\bot} message in $u.Ch$ then $u$ will send a \tempdelegate{v} message to itself upon receipt.
\item There can be no \introduce{v,w} for some node $w$ in $u.Ch$, since (i) if $id(w)<id(u)$, then due to Invariant~3 $v \in Right(w)$ which contradicts the choice of $u$ and (ii) if $id(w)>id(u)$ then according to the pseudocode $w$ can only send \introduce{u,w} message to $v$ and not vice-versa.
\item If there is a \linearize{v} message in $u.Ch$, then $u$ will either convert it into a \tempdelegate{v} message to itself or delete a previously saved reference to $v$ and send an \tempdelegate{v} to a node with a higher id.
\item If there is a \tempdelegate{v} message in $u.Ch$, $u$ either saves the reference (thereby deleting the implicit edge) or sends a \tempdelegate{v} to a node to a node with a higher id.
\item If there is \revandlin{} message in $u.Ch$ (i.e., $v \in nodeList$), then $u$ either saves the reference or sends  a \tempdelegate{v} to itself.
\item There can be no \revandlinACK{} message in $u.Ch$  (since $u$ is staying).
\end{itemize}
Consider the case in which $(u,v)$ is explicit.
If there is no node $x \in Right(u)$ with $id(u) < id(x) < id(v)$, then $u$ will introduce itself to $v$ in \timeout. 
The leaving node $v$ saves the reference of $u$ and sends \revandlinREQ{} to $u$, and according to the protocol $u$ will eventually delete its reference to $v$ due to \revandlin{} message.
If there exists a $x \in Right(u)$ with $id(u) < id(x) < id(v)$, then by definition of $v$ the node $x$ is staying.
In \timeout $u$ will send an \introduce{v,u} message to $x$, $x$ will respond to $u$ with a \linearize{v} message causing $u$ to remove $v$ from $Right(u)$.
Thus, there will be a state in which $u$ will never have an explicit or implicit reference to $v$ again.

Similar to the case that an edge $(u,v)$ does not exist, $u$ could always possibly get the reference of $v$ in an \linearize{v} or in an \introduce{v,\bot} message (i.e., we cannot exclude that nodes to the right of $v$ send these). 
However, the \linearize{v} message has to be a response to a former \introduce{v,u} message by $u$ according to Invariant~5, which are only sent by $u$ if it still has the reference to $v$.
Moreover, the \introduce{v,\bot} was sent by a node $w\neq v$ as a response to a former \linearize{v} message, according to Invariant~4.
Again, this can only happen if $id(v)>id(u)>id(w)$  $id(v)<id(u)<id(w)$ (i.e., it never happens). 

\end{proof}

\begin{lemma}
\label{lem:leaving_node_stops_pointing}
Consider a state $s_3$ of $CS_3$ and let $u$ and $v$ be leaving nodes with $id(u)<id(v)$.
If it holds in $s_3$ that (i) there is no edge $(u',u) \in NG$ with $id(u')<id(u)$, (ii) for any leaving node $v'$ with $id(u)<id(v')<id(v)$ there will never be an edge $(u,v') \in NG$ in a subsequent state, and (iii) there exists a $(w,u) \in NG$ with $w$ leaving and $id(u)<id(w)$,
  then there is a state $s'$ in $CS_3$ such that for the computation suffix $CS'$ starting in $s'$ it holds that $(u,v) \notin NG$ for every state in $CS'$.
\end{lemma}

\begin{proof}
Since there is no edge $(u',u) \in NG$ with $id(u')<id(u)$, no node to the left of $u$ can add a message to $u.Ch$.
Additionally, since $id(u)<id(v)$ no node $x$ to the right of $u$ can add a \tempdelegate{v} or \introduce{v,x} message to $u.Ch$, according to the protocol.
Furthermore, no node to the right of $u$ can send a \linearize{v} to $u$, since the message has to be a response to a former \introduce{v,u} message by $u$ (according to Invariant~5), which $u$ does not send.
Moreover, no node to the right of $u$ can send an \introduce{v,\bot}, since it is has to be sent by a node $w\neq v$ as a response to a former \linearize{v} message, according to Invariant~4.
This can only happen if $id(v)>id(u)>id(w)$ or $id(v)<id(u)<id(w)$ (i.e., it never happens).
Finally, no leaving node to the right can add a message to $u.Ch$ that contains the reference of $v$, because for any leaving node $v'$ with $id(u)<id(v')<id(v)$ there will never be an edge $(u,v') \in NG$ and Invariant~7.
Note that we only consider states in $CS_3$, therefore the above mentioned invariants hold.

At first assume that no edge $(u,v)$ exists.
Analogous to the same situation in Lemma~\ref{lem:staying_node_stops_pointing}, one can show that statement of the lemma is true.

In case $(u,v)$ exists, $(u,v)$ can be a multi-edge and be explicit as well as implicit. 
At first consider all implicit edges $(u,v)$.
\begin{itemize}
\item If there is an \introduce{v,\bot} message or \introduce{v,x} message in $u.Ch$ for some node $x$, then $u$ will save the reference of $v$.
\item In case there is a \tempdelegate{v} message in $u.Ch$, $u$ either saves the reference (thereby deleting the implicit edge) or sends an \tempdelegate{v} to a node with a higher id.
\item If there is a \revandlin{} message in $u.Ch$, it cannot contain the reference of $v$, since for the leaving sender $id(u)<id(sender)<id(v)$ has to hold (contradicting the choice of $(u,v)$ and the fact that for any leaving node $v'$ with $id(u)<id(v')<id(v)$ there will never be an edge $(u,v') \in NG$).
\item There can be no \revandlinACK{v,uniqueValue} message in $u.Ch$ that contains $v$ (since $id(u)<id(v)$ and Invariant~6).
\end{itemize}

Therefore eventually, $(u,v)$ is only an explicit edge.
Due to our choice of $u,v,w$ in the statement the node $w$ eventually sends a \revandlinREQ{right} to $u$ and $u$ responds with a \revandlin{u,uniqueValue} to $v$.
Node $v$ will receive said message, save $u$ in its local memory and send a \revandlin{nodeList,uniqueValue} back to $u$.
Consequently, $u$ deletes its reference to $v$ and saves the $nodeList$ instead.
Note that any further \revandlin{nodeList,uniqueValue} message from $v$ do not create an edge $(u,v)$, since $u$ has no node $x$ in its local memory with $uniqueValue[x]=uniqueValue$, so it only saves the nodeList itself.
Thus, there is a $s'$ in $CS_3$ such that for the computation suffix $CS'$ starting in $s'$ it holds that $(u,v) \notin NG$ for every state in $CS'$.
\end{proof}

 With these two lemmas in place, we can focus on the main statement.
 Note that since $CS_3$ is a computation suffix of $C$, by our initial assumption there exists at least one present leaving node in $CS_3$.
  Consider the set $L$ of present leaving nodes $x$ with the property that throughout $CS_3$ there does not exist a leaving node $y$ with $id(y) > id(x)$ with $x \in Left(y)$ or $x \in \templeft[y]$.
  Furthermore, let $u^*$ be the node with minimum id in $L$.
  Such a node must always exist since due to Lemma~\ref{lem:left_and_right_are_what_they_say}, the present leaving node with highest id is always in $L$.

We will show a contradiction to our initial assumption by proving that the node $u^*$ can execut e\textbf{exit} eventually. In order to do so consider the following lemma.

\begin{lemma}
There is a computation suffix $CS^*$ of $CS_3$ such that no edge $(u,u^*)$ with $id(u)<id(u^*)$ exists in $CS^*$.
\end{lemma}

\begin{proof}
We will prove the statement by induction over all leaving nodes $v$ with $id(v) \leq id(u^*)$.
For the sake of simplicity we address those nodes by $v_1, v_2, \ldots v_k =u^*$ with $id(v_i)<id(v_{i+1})$

For the induction base consider the leaving node with lowest id $v_1$.
Let $w_1, \ldots,w_m$ with $id(w_i)<id(w_{i+1})$ be all nodes with a lower id than $v_1$.
By definition all $w_i$ nodes are staying.
Due to the definition of $v_1$ and $w_1$ Lemma~\ref{lem:staying_node_stops_pointing} is applicable (in fact part (ii) of the if-statement is irrelevant) and there is a suffix such that $(w_1,v_1)$ will cease to exist forever.
Consequently, Lemma~\ref{lem:staying_node_stops_pointing} is applicable to $w_2$ and we can continue this approach until we have a suffix such that no edge $(u,v_1)$ with $id(u)<id(v_1)$ exists in that suffix.

For the induction step assume that the statement holds for some leaving node $v_i$.
Similar to the induction base let $w_1, \ldots,w_\ell$  be all nodes with a lower id than $v_i$ and let $w_{\ell+1}, \ldots,w_m$ be all nodes with an id bigger than $v_i$ but smaller than $v_{i+1}$ (with $id(w_i)<id(w_{i+1})$).
At first consider all $w_i \in \{w_1, \ldots,w_\ell \}$ in increasing order. 
In case the currently considered node $w_i$ is staying, we can apply Lemma~\ref{lem:staying_node_stops_pointing} to show that there is a suffix such that all nodes with an id lower than $w_i$ will never have an edge to $v_{i+1}$.
In case the currently considered $w_i$ is leaving we can apply Lemma~\ref{lem:leaving_node_stops_pointing} to get the same outcome.
Now consider $v_i$, by the induction hypothesis, we know that we can also apply Lemma~\ref{lem:leaving_node_stops_pointing}.
For all $w_i \in \{w_{\ell+1}, \ldots,w_m \}$ we know that they are staying, i.e., Lemma~\ref{lem:staying_node_stops_pointing} is applicable again.
Therefore the induction step is complete which proves the statement.
\end{proof}
Aisde from this, we can show that there is also a computation suffix in which there exists no edge $(u,u^*)$ with $id(u)>id(u^*)$.
We can do so by an argument analogous to Lemma~\ref{lem:staying_node_stops_pointing} (only that in this case the staying nodes has a higher id) and due to the choice of $u^*$ (i.e., throughout the computation suffix $CS_3$ only staying nodes with higher id have an edge to $u^*$).

Consequently, there exists a state in $CS^*$ (and thereby also in $C$) such that for all nodes $u$ no edge $(u,u^*)$ exists. Therefore, $u^*$ cannot receive any messages anymore and once its channel is empty, \nidec evaluates to true (i.e., it executes \textbf{exit}). This is a contradiction to the choice of $C$.

 \end{proof}

\subsection{\blpp solves the linearization problem}
\label{sec:blpp_linearization_proof}
Here, we show the following theorem.

\begin{theorem}\label{thm:blpp_solves_linearization}
 \blpp is a self-stabilizing solution to the linearization problem.
\end{theorem}

\begin{proof}
Note that by Lemma~\ref{lem:fdp:liveness}, in every computation of \blpp there is a suffix in which all leaving nodes are gone.
Note that starting from this state, \blpp acts exactly as \blp.
By Lemma~\ref{lem:fdp:safety}, $NG$ is still weakly connected in this state.
Thus, the properties of Theorem~\ref{thm:blp_solves_linearization} are fulfilled, yielding that \blpp is a solution to the linearization problem as well. 
\end{proof}

\subsection{\blpp satisfies non-trivial monotonic searchability}
\label{sec:blpp_monotonic_searchability_proof}
Finally, we prove the following thereom concerning monotonic searchability.

\begin{theorem}\label{thm:blpp_guarantees_monotonic_searchability}
 \blpp admissible-message satisfies non-trivial monotonic searchability according to \srpp.
\end{theorem}

In general, the proof follows the structure of the results from Subsection~\ref{sec:monotonic_searchability_proof}.
However, since we want to satisfy monotonic searchability even under the presence of leaving nodes, the proof is more involved.
First we define $R_s(v)$ as the set of all staying nodes $x$ with $id(v) < id(x)$ for which there is a directed path from $v$ to $x$ consisting solely of explicit edges $(y,z)$ with $id(y) < id(z)$ that arise from $z \in Right(y)$.
Furthermore, we define $\rsvw := \{x \in R_s(v) | id(x) \leq id(w)\}$.
In addition, we define $L_s(v)$ as the set of all staying nodes $x$ with $id(x) < id(v)$ for which there is a directed path from $v$ to $x$ consisting solely of explicit edges $(y,z)$ with $id(z) < id(y)$ that arise from $z \in Left(y)$.
For a set of nodes$U$, we define $R_s(U) := U \cup \bigcup_{u \in U}R_s(u)$ and $L_s(U) := U \cup \bigcup_{u \in U}L_s(u)$.
Additionally, we define $R_s(U,ID) := \{x \in R_s(U) | id(x) \leq ID\}$, and $L_s(U,ID) := \{x \in L_s(U) | id(x) \geq ID\}$
Last, we have $\rsp(u) := R_s(u)$ if $u$ is leaving, or $\rsp := R_s(u) \cup \{u\}$ if $u$ is staying (with $\lsp(u)$ defined analogously).

Moreover, we define the following message invariants:
\begin{enumerate}
    \item If there is an \introduce{v,w} message with $w \neq \bot$ in $u.Ch$, then $v \neq w$, and $\rsp(u) \subseteq R_s(w)$ (or $\lsp(u) \subseteq L_s(w)$).
    \item If there is a \linearize{v} message in $w.Ch$, then there is a node $u \neq v$ with $u \in Right(w)$ and $\rsp(v) \subseteq R_s(u)$ if $w < v$ (or $u \in Left(w)$ and $\lsp(v) \subseteq L_s(u)$ if $v < w$).
    \item If there is a \revandlinACK{v,uniqueValue} message in $u.Ch$, then $u \neq v$ and $u.uniqueValues[v] = uniqueValue$ and $v$ is the only node with $u.uniqueValues[v] = uniqueValue$.
    \item If there is a \revandlin{nodeList,uniqueValue} message in $u.Ch$, then there is exactly one node $v$ with $u.uniqueValues[v] = uniqueValue$.
    Furthermore, $v$ is leaving, and $R_s(v) = R_s(nodeList)$ if $u < v$ (or $L_s(v) = L_s(nodeList)$ if $v < u$).
    \item If there is a \forwardprobe{source,destID,Next,seq} message in $u.Ch$, then
    \begin{enumerate}
	\item $id(source) < destID$ and $\forall x \in Next: id(x) \geq id(u)$ and $u = argmin_u\{id(u) | u \in Next\}$ 
	(alternatively $destID < id(source)$ and $\forall x \in Next: id(x) \leq id(u)$ and $u = argmax_u\{id(u) | u \in Next\}$).
	\item $id(source) < destID$ and $R_s(Next) \subseteq R_s(source)$ (or $destID < id(source)$ and $\lsp(u) \subseteq L(source)$).
	\item if $v$ exists with $id(v) = destID$ and $v$ is staying, such that $id(source) < destID$, and $v \notin R_s(Next,destID)$ (or $id(source) < destID$ and $v \notin L_s(Next,destID)$) then for every admissible state with $source.seq[destID] < seq$, $v \notin R_s(source,destID)$ ($v \notin L_s(source,destID)$).
    \end{enumerate}
    \item If there is a \psuccess{destID, seq, dest} message in $u.Ch$, then $id(dest) = destID$ and $dest \in R_s(u)$ if $destID > id(u)$ (or $dest \in L_s(u)$ if $destID < id(u)$), or $dest$ is leaving.
    \item If there is a \pfail{destID, seq} message in $u.Ch$, then either there is no staying node with id $destID$, or for every admissible state with $u.seq[destID] < seq$, $v \notin R_s(u)$ (and $v \notin L_s(u)$), where $v$ is the node with $id(v)= destID$.
    \item If there is a \search{v, destID} message in $u.Ch$ and $u$ is staying, then $id(u) = destID$ and $u \in R_s(v)$ if $id(v) < destID$ (or $u \in L_s(v)$ if $destID < id(v)$).
\end{enumerate}
A state is therefore admissible if all four invariants hold.
As in Section~\ref{sec:monotonic_searchability_proof}, we can prove:
\begin{lemma}\label{lem:blpp_once_admissible_always_admissible}
 If in a computation of \blpp, there is an admissible state, then all subsequent states will be admissible as well.
\end{lemma}
The general structure of the proof is similar to the proof of Lemma\ref{lem:once_admissible_always_admissible}, although the details are different as we have to take into account that nodes can become leaving and due to the additional message invariants.

First, we show the following:
\begin{lemma}\label{lem:blpp_once_first_four_invariants_hold_then_always}
 If in a computation of \blpp, there is a state in which Invariants~1-4 hold, then in all subsequent states Invariants~1-4 will hold.
\end{lemma}
\begin{proof}
 Assume there is a state $s_1$ in which Invariant~1-4 hold, such that in the (direct) subsequent state $s_2$ one of the Invariants~1-4 does not hold.
 First of all, check that none of the first four invariants can be invalidated because some node becomes leaving.
 Secondly, note that the first four invariants cannot become falsified due to a new \introduce{v,w} or \linearize{v} message for very similar reasons as in the proof of Lemma~\ref{lem:once_first_two_invariants_hold_then_always} (since in this part \blp and \blpp are exactly the same).
 Furthermore, note that according to the protocol when a node $w$ sends a \revandlinACK{v,uniqueValue} to a node $u$, then $w=v$ and it makes sure that $uniqueValue$ is stored in $v.uniquevalues[u]$ (and we assume that $uniqueValue$ is only stored for $u$).
 Thus, sending such a message also cannot invalidate one of the first four invariants.
 Moreover, note that when a node $v$ sends a \revandlin{nodeList,uniqueValue} message to a node $u$ with $u < v$ between state $s_1$ and $s_2$, then $v$ must have received a \revandlinACK{u,uniqueValue} message right before and $v$ must be leaving.
 Since Invariant~3 holds in $s_1$, this means that $u.uniqueValues[v] = uniqueValue$ and $v$ is the only node such that $u.uniqueValues[v] = uniqueValue$.
 In addition, when sending the message, $v$ added all nodes from $Right(u)$ to $nodeList$.
 Thus, in state $s_2$, $R_s(v) = R_s(nodeList)$ holds and $v$ is the only node with $uniqueValues[v] = uniqueValue$.
 If $v < u$, $L_s(v) = L_s(nodeList)$ holds, for analogous arguments.
 Besides, note that the $R_s(v) = R_s(nodeList)$ part of Invariant~4 for a node $v$ cannot be invalidated due to the addition of any node to the set $Right(v)$ (or $Left(v)$) because $v$ is leaving and a leaving node never adds a member to $Right$ (or $Left$).
 Any other addition of a node to a set $Right(x)$ (or $Left(x)$) for another node $x$ adds this node to $R_s(v)$ and $R_s(nodeList)$ at the same time or not at all.
 
 Thus, the only event that can invalidate one of the first four invariants is the removal of a node $y$ from a set $Right(x)$ or $Left(x)$ for a node $x$.
 This may only happen in a \linearize{y} action for a staying node or a \revandlin{nodeList, uniqueValue} action.
    We will consider both actions invidivdually.
    
    First of all, assume a \linearize{y} action has been executed in a staying node $w$ between $s_1$ and $s_2$ and thus removed a node $y$ from $Right(w)$ (or $Left(w)$).
    This can only happen if there was a \linearize{y} message in $w.Ch$ in $s_1$ for which, by definition of $s_1$, Invariant~2 holds.
    Thus, there is a node $u \neq y$ with $u \in Right(w)$ and $\rsp(y) \subseteq R_s(u)$ (or $u \in Left(w)$ and $\rsp(y) \subseteq L_s(v)$), implying that after the removal of $(w,y)$, $\rsp(u) \subseteq R_s(w)$ ($\lsp(u) \subseteq L_s(w)$) still holds, i.e., there is no node $x$ for which a node has been removed from $R_s(x)$ and the first four invariants cannot be invalidated due to the change.
 
    Now assume that a \revandlin{nodeList, uniqueValue} action has been executed in a node $u$ between state $s_1$ and $s_2$.
    In this case, the corresponding message must have been in $u.Ch$ in $s_1$.
 Since in $s_1$ the first four invariants hold, by the fourth invariant, there must be exactly one node $v$ that is leaving with $u.uniqueValues[v] = uniqueValue$, and $R_s(v) = R_s(nodeList)$ if $u < v$ (or $L_s(v) = L_s(nodeList)$, otherwise).
 W.l.o.g. assume that $u < v$ (note that in case $v < u$ and $u$ leaving, no node is removed from or added to $Left(u)$ at all, but in this case, the invariant still holds, which is what we want to prove anyway).
 If $v \notin Right(x)$, no node is removed from or added to $Right(x)$ at all and the claim follows immediately.
 Thus, assume $v \in Right(x)$.
 In this case, $u$ removes $v$ from $Right(u)$ and adds $nodeList$ to $Right(u)$.
 Since $R_s(v) = R_s(nodeList)$ and $R_s(v) \subseteq R_s(u)$, and $v \notin R_s(v)$ (because $v$ is leaving), no node has been removed from or added to $R_s(u)$ after the action has been performed, implying that all four invariants still hold.  
\end{proof}
Similar to Lemma~\ref{lem:R_grows_monotonically}, one can show the following:
\begin{lemma}\label{lem:blpp_Rs_grows_monotonically}
  If there is a state in which the first four invariants hold, and $\rsp(x) \subseteq R_s(v)$ ($\lsp(x) \subseteq L_x(v)$), then in every subsequent step, $\rsp(x) \subseteq R_s(v)$ ($\lsp(x) \subseteq L_s(v)$).
\end{lemma}
\begin{proof}
Assume there is a state $s_1$ such that $\rsp(x) \subseteq R_s(v)$ holds, but in the (direct) subsequent state $s_2$, $\rsp(x) \subseteq R_s(v)$ does not hold.
We consider all possible reasons for why $\rsp(x) \subseteq R_s(v)$ does not hold in $s_2$.
Obviously, neither the addition of a node to $R_s(v)$ nor the removal of a node from $\rsp(x)$ can violate the claim.
Note that if a node $z$ is added to $\rsp(x)$, this happens because a node $y \in \rsp(x)$ added $z$ to $Right(x)$.
However, since $y \in R_s(v)$, $z$ is also added to $R_s(v)$ (by definition of this set).
This yields that the only reason for the claim to be incorrect in $s_2$ is that a (staying) node $z \in \rsp(x)$ was removed from $R_s(v)$ but not from $\rsp(x)$.
We consider all possible cases for this.

First, assume $z$ was removed from $R_s(v)$ because $z$ became leaving.
Then $z$ was also removed from $\rsp(x)$.

Secondly, assume that $z$ was removed from $R_s(v)$ due to a \linearize{y} action at a node $w \in R_s(v) \cup \{v\}$ between $s_1$ and $s_2$.
Then, by the second invariant, there was a node $u \neq y$ with $u \in Right(w)$ and $\rsp(y) \subseteq(u)$ in $s_1$.
Thus, after $y$ is removed from $Right(w)$, $\rsp(y) \in R_s(w)$ still holds, implying $\rsp(y) \subseteq R_s(v)$, i.e., neither $y$ nor any other node $z$ in $R_s(y)$ was removed from $R_s(v)$.

Thirdly, assume a staying node $z \in \rsp(x)$ was removed from $R_s(v)$ but not from $\rsp(x)$ due to a \revandlin{nodeList,uniqueID} action in a node $u$, removing node $y$ from $Right(u)$.
In this case, according to Invariant~4, $y$ is the unique node with $u.uniqueVAlues[y] = uniqueValue$, $y$ is leaving, and $R_s(y) = R_s(nodeList)$.
Thus, when $y$ is removed from $Right(u)$ and $NodeList$ is added to $Right(u)$, no node is removed from $R_s(u)$, implying that no node is removed from $R_s(v)$.

Thus, the claim holds in every case.
Note that the argument for $\lsp(x) \subseteq L_x(v)$ is completely analogous.
\end{proof}
Using this, we can prove the following lemmata:
\begin{lemma}\label{lem:blpp_once_first_five_invariants_hold_then_always}
 If in a computation of \blpp, there is a state in which Invariants~1-5 hold, then in all subsequent states Invariants~1-5 will hold.
\end{lemma}
\begin{proof}
 Assume there is a state $s_1$ in which Invariant~1-5 hold, such that in the (direct) subsequent state $s_2$ one of the first five invariants does not hold.
 By Lemma~\ref{lem:blpp_once_first_four_invariants_hold_then_always}, this can only be Invariant~5.
 Note that Invariant~5a) is equal to Invariant~3a) from Section~\ref{sec:monotonic_searchability_proof}.
 Thus, Invariant~5a) cannot be violated for the same reasons mentioned in the proof of Lemma~\ref{lem:once_first_three_invariants_hold_then_always}.

 Note that if Invariant~5b) and 5c) hold for a \forwardprobe{source,destID,Next,seq} message when this message is sent, they also do so when the message is delivered because of Lemma~\ref{lem:blpp_Rs_grows_monotonically}.
 Thus, the only reason why Invariant~5b) or 5c) do not hold in $s_2$ is that a new \forwardprobe{source,destID,Next,seq} message has been sent.
 There may be two reasons for this:
 Either because a node $u$ executed \timeout, or because a node $u$ received another \forwardprobe{source,destID,Next',seq} message.
 We consider both cases individually (each time, for $id(source) < destID$ because the other case is analogous).
 
 In the first case, the \forwardprobe{source,destID,Next,seq} message is sent to $u$ itself, with $u = source$ and $Next=\{u\}$, which is why Invariant~5b) holds.
 Also note that since $u.seq[destID]$ is monotonically increasing, and $seq = source.seq[destID]$ in this state, if there was an admissible state with $source.seq[destID] < seq$ with $v \in R_s(source,destID)$, then this must have been a previous state.
 Note that $v \in R_s(source,destID)$ implies $\rsp(v) \subseteq R_s(source)$.
 By Lemma~\ref{lem:blpp_Rs_grows_monotonically}, $\rsp(v) \subseteq R_s(source)$ must still hold in $s_1$, which, if $v$ is staying, implies $v \in R_s(source,destID)$.
 Thus, Invariant~5c) still holds in this case.
 
 In the second case, Invariant~5 held for the \forwardprobe{source,destID,Next',seq} message $u$ received.
 Note that $u$ only sends the \forwardprobe{source,destID,Next,seq} message if $id(u) \neq destID$.
 Thus, if there is a $v$ such that $id(v) = destID$ then $u \neq v$ and since $R_s(Next,destID)$ and $R_s(Next',destID)$ only differ in $u$ (since $Next = Next'\setminus\{u\} \cup Right(u)$), Invariant~5c) also holds for the new message. 
 Notice that the new message is sent to a node $w \in Right(u)$ or $w \in Next'$, i.e., $w \in R_s(Next')$ in any case.
 $R_s(Next') \subseteq R_s(source)$ implies $R_s(Next) \subseteq R_s(source)$ ($Next = Next'\setminus\{u\} \cup Right(u)$), yielding the claim of Invariant~5b) for the new message.
 
 All in all, Invariant~5 has to hold in $s_2$, too, proving the claim.
\end{proof}
\begin{lemma}\label{lem:blpp_once_first_seven_invariants_hold_then_always}
 If in a computation of \blpp, there is a state in which Invariants~1-7 hold, then in all subsequent states Invariants~1-7 will hold.
\end{lemma}
\begin{proof}
  Again, assume there is a state $s_1$ in which Invariant~1-7 hold, such that in the (direct) subsequent state $s_2$ one of the first seven invariants does not hold.
 By Lemma~\ref{lem:blpp_once_first_five_invariants_hold_then_always}, this can only be Invariant~6 or Invariant~7.
 Observe that Invariant~6 and Invariant~7 can only be violated if there is a node $v$ with $id(v) = destID$ and $v$ is staying.
 Again, by Lemma~\ref{lem:blpp_Rs_grows_monotonically} and because $\rsp(x) \subseteq R_s(y)$ is equivalent to $x \in R_s(y)$ if $x$ is staying, any of the two invariants can only be violated if a new \psuccess{destID,seq,dest} or a new \pfail{destID,seq} was sent by a node $w$ between $s_1$ and $s_2$.
 We consider both cases individually.
 
 Assume a new \psuccess{destID,seq,dest} message has been sent by $w$ to a node $u$.
 	According to the protocol, this only happens in a \forwardprobe{} action, when a \forwardprobe{source,destID,Next,seq} message has arrived at $w = dest$ with $id(w) = destID$ and $u = source$.
 	As stated before, $w$ must be staying.
	Thus, Invariant 5~b) implies $dest \in R_s(u)$.
	
	For the \pfail{} messages, assume a node $w$ sends a \pfail{destID, seq} message to a node $u$.
	According to the protocol, this only happens in a \forwardprobe{} action, when a \forwardprobe{source,destID,Next,seq} message has arrived at $w$ with $id(w) \neq destID$, $u = source$ and $Next = \{w\}$ and there is no $y$ in $Right(x)$ with $id(y) \leq destID$.
	If no staying node with id $destID$ exists, we are done.
	Otherwise, we have that for this node $v$, $v \notin R(Next,w)$.
	By Invariant~3c), this implies the claim.

 Thus, Invariant~6 and Invariant~7 have to hold in $s_2$, too, proving the claim.
\end{proof}
Now we can finally prove Lemma~\ref{lem:blpp_once_admissible_always_admissible}:
\begin{proof}
 Assume there is an admissible state $s_1$, such that the (direct) subsequent state $s_2$ is not admissible.
 By Lemma~\ref{lem:blpp_once_first_seven_invariants_hold_then_always}, only Invariant~8 can be violated in $s_2$.
 However, by a similar argument as in the proof of Lemma~\ref{lem:once_admissible_always_admissible}, this is not possible.
\end{proof}
The following also holds:
\begin{lemma}\label{lem:blpp_admissible_state_always_exists}
 In every computation of \blpp there is an admissible state.
\end{lemma}
\begin{proof}
 Note that according to Lemma~\ref{lem:fdp:liveness}, every computation of \blpp has is a suffix in which nodes that will eventually be leaving are gone and note that these nodes do not perform any actions.
 Furthermore, note that a \revandlin{nodeList,uniqueValue} message is only sent if a node received a \revandlinACK{v,uniqueValue} message.
 Moreover, a \revandlinACK{v,uniqueValue} message can only be sent if a node receives a \revandlinREQ{DIR} message.
 Such a message, can only be sent from a leaving node.
 However, in the aforementioned suffix, no leaving node can send a message any more. 
 Thus, there is a suffix, in which the third and the fourth invariant always hold.
 
 Note that by Theorem~\ref{thm:blpp_solves_linearization}, the remaining nodes will converge to the list.
 In this state, similar to the argument used in the proof of Lemma~\ref{lem:admissible_state_always_exists}, no new \introduce{v,w} messages with $v \neq w$ and no new \linearize{u} messages can be initiated, i.e., the first two invariants always hold.
 Note that since $\rsp(v) \subseteq R_s(u)$ is equivalent to $v \in R(u)$ if $v$ is staying, and the system only consists of staying nodes in the current suffix, Invariant~5-8 are equivalent to Invariant~3-6 in Section~\ref{sec:monotonic_searchability_proof}.
 Thus, the rest of the proof is analogous to the proof of Lemma~\ref{lem:admissible_state_always_exists}.
\end{proof}
Note that Lemma~\ref{lem:blpp_once_admissible_always_admissible} and Lemma~\ref{lem:blpp_admissible_state_always_exists} imply the following corollary:
\begin{corollary}
    In every computation of \blpp, there exists a suffix in which every state is admissible. 
\end{corollary}
For the rest of this subsection, we assume that every computation starts in an admissible state.
This is due to the fact that monotonic searchability must hold starting from admissible states only.
Furthermore, w.l.o.g. we only consider requests \search{u,destID} with $id(u) < destID$.

As in Section~\ref{sec:monotonic_searchability_proof}, we need some additional results before we can prove Theorem~\ref{thm:blpp_guarantees_monotonic_searchability}.
\begin{lemma}\label{lem:blpp_forwardprobe}
Assume for a \forwardprobe{v,destID, Next, seq} message $m \in x.Ch$, there is a $u \in R_s(Next,destID)$.
Then either $u = x$ or there will be a state in which a \forwardprobe{v,destID, Next', seq} message is in $y.Ch$ for some node $y$ with $id(y) > id(x)$ and $u \in R_s(Next',destID)$, or $u$ is leaving.
\end{lemma}
\begin{proof}
    Assume $u \neq x$.
    Note that when $m$ is received by $x$, a new message with $Next' = Next\setminus\{x\} \cup Right(x)$ will be sent.
    According to the fifth invariant, for all nodes $z$ in $Next$, $id(z) > id(y)$ holds, and $x$ is the node with minimum id among all nodes in $Next$.
    By Lemma~\ref{lem:left_and_right_are_what_they_say}, the same holds for the nodes $z$ in $Right(x)$.
    Thus, $x$ is the node with minimum id among all ones in $R(Next,w)$ and for the node $y$ to which a new \forwardprobe{v,destID,Next',seq} message is sent it holds $id(y) > id(x)$.
    Furthermore, $R_s(Next(x),destID)\setminus \{x\} \subseteq R_s(Next',destID)$ implying $u \in R_s(Next',destID)$ unless $u$ has become leaving.
\end{proof}
This allows us to prove the following lemma:
\begin{lemma}\label{lem:blpp_w_in_R_leads_to_message_success}
 For every message $m = \forwardprobe{v,destID, Next, seq} \in u.Ch$ with $id(u) < destID$, it holds that if there is a staying node $w$ with $id(w)=destID$ in the network and $w \in R_s(u)$, then eventually there will be a \forwardprobe{v,destID,Next'} message in $w.Ch$, or $w$ will be leaving.
\end{lemma}
\begin{proof}
	Note that when $m$ arrives as $u$, $Next$ will be changed such that $R_s(u,w) \subseteq R_s(Next,w)$.
	If $w \in R_s(u)$, then $w \in R_s(Next,w)$ afterwards.
	Thus, by applying Lemma~\ref{lem:blpp_forwardprobe} recursively, we have that eventually a \forwardprobe{v,destID, Next', seq} will be in $w.Ch$, which will be received according to the fair message receipt assumption, unless $w$ becomes leaving.
\end{proof}

Using these results, the proof of Theorem~\ref{thm:blpp_guarantees_monotonic_searchability} is analogous to the proof of Theorem~\ref{thm:blp_guarantees_monotonic_searchability} (substituting $R(v)$ by $R_s(v)$, noting that $\rsp(v) \subseteq R_s(u)$ is equivalent to $v \in R(u)$ if $v$ is staying, and using Lemma~\ref{lem:blpp_Rs_grows_monotonically} instead of Lemma~\ref{lem:R_grows_monotonically}, and Lemma~\ref{lem:blpp_w_in_R_leads_to_message_success} instead of Lemma~\ref{lem:w_in_R_leads_to_message_success}).
Note that as soon as a node becomes leaving, searchability to this node does not need be satisfied any longer.

\section{Conclusion and Outlook}
To the best of our knowledge, we presented the first protocol that self-stabilizes a topology whilst satisfying monotonic searchability.
We focused on the line topology as a starting point and extended our protocol such that it additionally solves the Finite Departure Problem.
In the design of our protocol, it turned out that the principle of delegating explicit edges only if they have been successfully introduced before is crucial to enable monotonic searchability.
A natural open question is whether the application of this principle is sufficient for monotonic searchability.
That is, does applying this principle to other protocols that stabilize a topology (e.g., rings, skip-graphs, Delaunay graphs) directly yield monotonic searchability, or do other topologies require more-specialized solutions?

%

\bibliography{bibliography.bib}

\end{document}